\documentclass[cmp]{svjour}  %envcountsame
\usepackage[utf8]{inputenc}

\usepackage{enumerate}
\usepackage{mathrsfs}
\usepackage{amssymb}
\usepackage{amsmath}
\usepackage{mathtools}
\usepackage{graphicx}
\usepackage{hyperref}

\graphicspath{ {./images/} }

\newcommand\norm[1]{\left\lVert#1\right\rVert}

\usepackage[english]{babel}
\usepackage{cite}
\usepackage{amsfonts}
\usepackage{stackrel}
\usepackage{color}

\newcommand{\bx}{{ \mathbf{x} }}

\newcommand{\bn}{{ \mathbf{n}}}
\newcommand{\br}{{\mathbf{r}}}

\newcommand{\bu}{{ \mathbf{u}}}
\newcommand{\bJ}{{ \mathbf{J}}}
\newcommand{\bv}{{\mathbf{u}}}

\newcommand{\bS}{{\mathbf{S}}}

\newcommand{\bI}{{\mathbf{I}}}

\newcommand{\bV}{{\mathbf{V}}}

\newcommand{\be}{\begin{equation}}
\newcommand{\ee}{\end{equation}}
\newcommand{\lb}{\label}

\newcommand{\bzed}{{\bf 0}}

\newcommand{\btau}{{\mbox{\boldmath $\tau$}}}
\newcommand{\biota}{{\mbox{\boldmath $\iota$}}}
\newcommand{\bomega}{{\mbox{\boldmath $\omega$}}}

\newcommand{\bphi}{{\mbox{\boldmath $\phi$}}}
\newcommand{\bpsi}{{\mbox{\boldmath $\psi$}}}

\newcommand{\btimes}{{\mbox{\boldmath $\,\times\,$}}}
\newcommand{\bdot}{{\mbox{\boldmath $\,\cdot\,$}}}
\newcommand{\bdots}{{\mbox{\boldmath $\,:\,$}}}
\newcommand{\grad}{{\mbox{\boldmath $\nabla$}}}

\newcommand{\cE}{\mathcal{E}}
\newcommand{\cT}{{\mathcal{T}}}
\newcommand{\cN}{{\mathcal{N}}}
\newcommand{\ext}{\textbf{Ext}}

\journalname{Communications in Mathematical Physics}

\begin{document}

\title{{Onsager Theory of Turbulence,\\
the Josephson-Anderson Relation,\\
and the D'Alembert Paradox}}
\titlerunning{Onsager, Josephson, D'Alembert}
\author{Hao Quan\inst{1} and Gregory L. Eyink \inst{1}\fnmsep\inst{2}}
\institute{Department of Applied Mathematics \& Statistics\\ The Johns Hopkins University, Baltimore, MD 21218, USA\\ \email{haoquan@jhu.edu} \and 
 Department of Physics and Astronomy \\ The Johns Hopkins University, Baltimore, MD 21218, USA\\
 \email{eyink@jhu.edu}}
\authorrunning{H. Quan \& G. Eyink}

\date{\today}
\communicated{???}

\maketitle
\begin{abstract}
The Josephson-Anderson relation, valid for the incompressible Navier-Stokes solutions which describe flow around 
a solid body, instantaneously equates the power dissipated by drag to the flux of vorticity across the flow lines of the 
potential Euler solution considered by d'Alembert. Its derivation involves a decomposition of the velocity field 
into this background potential-flow field and a solenoidal field corresponding to the rotational wake behind the body,
with the flux term describing transfer from the interaction energy between the two fields and into kinetic energy of the 
rotational flow. We establish the validity of the Josephson-Anderson relation for the weak solutions of the Euler 
equations obtained in the zero-viscosity limit, with one transfer term due to inviscid vorticity flux and 
the other due to a viscous skin-friction anomaly. Furthermore, we establish weak forms of the local 
balance equations for both interaction and rotational energies. 
We define nonlinear spatial fluxes of these energies and show 
that the asymptotic flux of interaction energy to the wall equals the anomalous skin-friction term in the 
Josephson-Anderson relation. However, when the Euler solution satisfies suitably the no-flow-through 
condition at the wall, then the anomalous term vanishes. In this case, we 
can show also that the asymptotic flux of rotational energy to the wall must vanish and 
we obtain in the rotational wake the Onsager-Duchon-Robert relation between viscous dissipation anomaly 
and inertial dissipation due to scale-cascade. In this way we establish a precise connection between 
the Josephson-Anderson relation and the Onsager theory of turbulence, and we provide a novel resolution of the d'Alembert paradox. 
\end{abstract}

\section{Introduction}

\subsection{Josephson-Anderson Relation for Flow Around a Solid Body}
The Josephson-Anderson relation was first derived for voltage-drops 
in superconductors \cite{josephson1965potential} and chemical potential differences in superfluids
\cite{anderson1966considerations}, relating these to transverse motion of quantized vortices.
A ``detailed relation'' of Huggins \cite{huggins1970energy} further connected 
superfluid dissipation to flux of vortices across streamlines of the 
background potential flow and was applied by Huggins also to classical turbulent channel flow 
described by the incompressible Navier-Stokes equation \cite{huggins1994vortex}. 

Recently, the detailed Josephson-Anderson relation was extended to incompressible 
Navier-Stokes solutions describing flow around a smooth, solid body \cite{eyink2021josephson}.
This result made a direct connection with the classical paradox raised by work of d'Alembert 
\cite{dalembert1749theoria,dalembert1768paradoxe}, who showed that the potential Euler solution 
$\bv_\phi=\grad\phi$ for flow around the body $B$ predicts no drag. See Fig.~\ref{fig1} for the 
context. Following the earlier work on superfluids,
\cite{eyink2021josephson} introduced a ``rotational velocity'' $\bv_\omega^\nu:=\bv^\nu-\bv_\phi$ 
which accounts for all flow vorticity $\bomega^\nu=\grad\btimes\bv^\nu$ and which satisfies the equation 
\begin{eqnarray}
\partial_t\bv_\omega^\nu 
&=& \bv^\nu\btimes\bomega^\nu-\nu\grad\btimes\bomega^\nu -
\grad\left(p_\omega^\nu+\frac{1}{2}|\bv_\omega^\nu|^2+\bv_\omega^\nu\bdot\bv_\phi\right), \cr
&=& -\grad\bdot (\bv_\omega^\nu\otimes\bv_\omega^\nu+\bv_\omega^\nu\otimes\bv_\phi+\bv_\phi\otimes\bv_\omega^\nu) 
-\grad p_\omega^\nu -\nu\grad\btimes\bomega^\nu, 
\lb{NS-omega-mom2} \end{eqnarray} 
expressing local conservation of the ``vortex momentum'' ${\bf P}_\omega^\nu=\rho\int_\Omega \bv_\omega^\nu\, dV.$ 
Here the pressure $p_\omega^\nu$ is to be determined by the 
divergence-free constraint $\grad\bdot\bv_\omega^\nu=0$ and $\rho$ is mass density of the fluid. 
The formulation \eqref{NS-omega-mom2} of the incompressible Navier-Stokes 
equation gives a precise mathematical description of the rotational wake that develops behind the body. 
In this setting, the Josephson-Anderson relation follows by considering balance equations for the 
interaction energy between potential and rotational flow, $E_{int}^\nu=\rho\int_\Omega \bv_\omega^\nu\bdot\bv_\phi\,dV,$ 
and the kinetic energy in the rotational flow itself, $E_\omega^\nu=(\rho/2)\int_\Omega |\bv_\omega^\nu|^2 \,dV,$
that is: 
\begin{eqnarray}
&& \partial_t\left(\bv_\phi\bdot \bv_\omega^\nu \right) +
\grad\bdot\left[\left(p_\omega^\nu+\frac{1}{2}|\bv_\omega^\nu|^2 +\bv_\omega^\nu\bdot\bv_\phi\right)\bv_\phi 
 +\left(p_\phi+\frac{1}{2}|\bv_\phi|^2\right) \bv_\omega^\nu\right] \cr
&& \qquad\qquad\qquad\qquad = +\bv_\phi\bdot(\bv^\nu\btimes\bomega^\nu-\nu\grad\btimes\bomega^\nu) \cr 
&& 
\lb{Eint-loc} \end{eqnarray} 
and 
 \begin{eqnarray}
&& \partial_t\left(\frac{1}{2}|\bv_\omega^\nu|^2\right) +
\grad\bdot\left[\left(p_\omega^\nu+\frac{1}{2}|\bv_\omega^\nu|^2 +\bv_\omega^\nu\bdot\bv_\phi\right)
\bv_\omega^\nu-
\nu\bv^\nu\btimes\bomega^\nu\right] \cr
&& \qquad\qquad\qquad = -\bv_\phi\bdot(\bv^\nu\btimes\bomega^\nu-\nu\grad\btimes\bomega^\nu)-\nu|\bomega^\nu|^2. 
\lb{Eom-loc} \end{eqnarray} 
\begin{figure}
\center
\includegraphics[width=.75\textwidth]{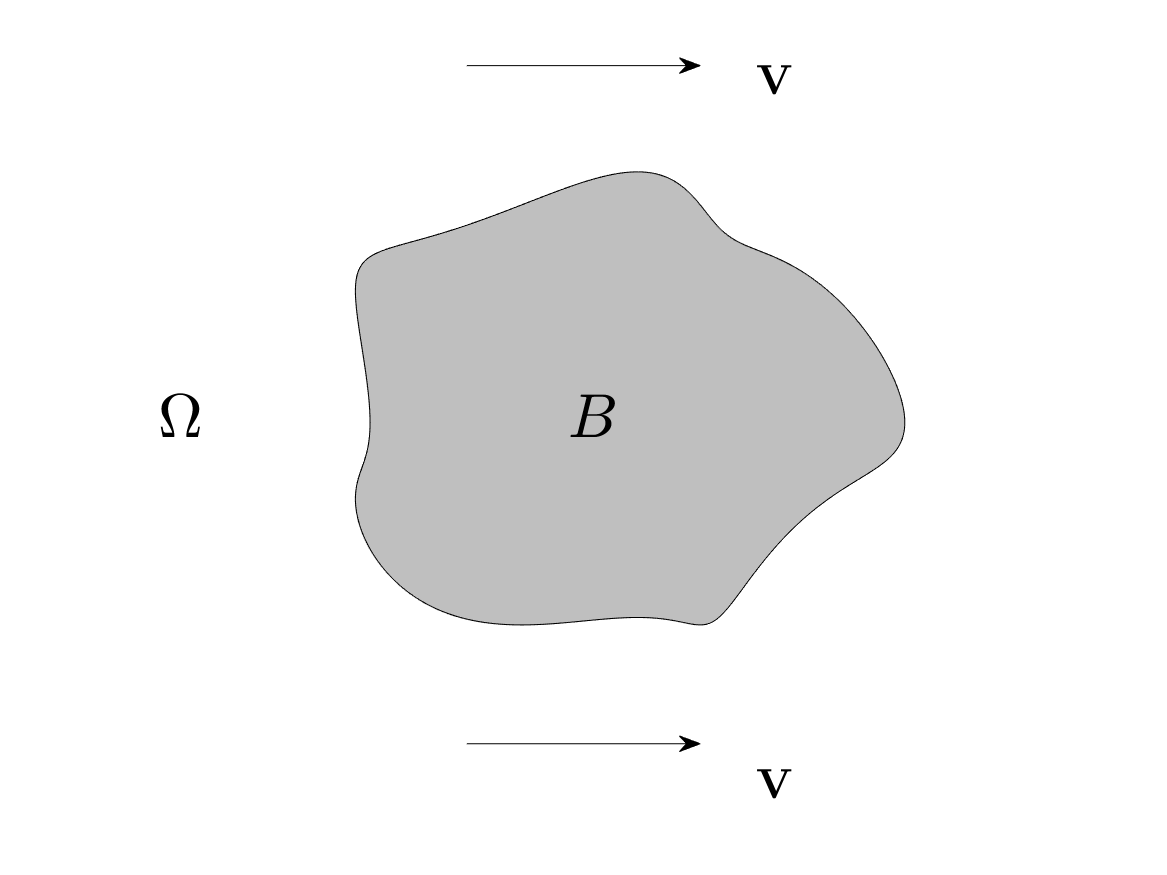}
\caption{Flow around a finite body $B$ in an unbounded region $\Omega$
filled with an incompressible fluid moving at a velocity ${\bf V}$ at far distances.} 
\label{fig1} \end{figure}
Clearly, the space-integrated expression 
\be {\mathcal T}^\nu= -\rho\int_\Omega \bv_\phi\bdot(\bv^\nu\btimes\bomega^\nu-\nu\grad\btimes\bomega^\nu)\,dV 
\lb{trans} \ee 
represents a transfer from the interaction energy to the energy of rotational flow,
which is thereafter disposed by viscosity, and it can be interpreted physically in terms 
of the flux of vorticity across the streamlines of the potential flow \cite{huggins1970energy,eyink2021josephson}. 
This quantity has furthermore a remarkable instantaneous relation to the drag force of the fluid 
acting on the body: 
\be {\bf F}^\nu= \rho\int_{\partial B} (-p_\omega^\nu\bn+\btau_w^\nu) dA \lb{force} \ee 
where viscous skin friction or wall stress is given by $\btau_w^\nu=2\nu\bS^\nu\bn=\nu\bomega^\nu\btimes\bn,$
$\bn$ is the unit normal on $\partial B$ pointing into $\Omega,$  and 
$S_{ij}^\nu=(1/2)(\partial_iu_j^\nu+\partial_ju_i^\nu)$ is the strain-rate tensor. The Josephson-Anderson 
relation derived in \cite{eyink2021josephson} states that the transfer \eqref{trans} is 
equal to $D^\nu:={\bf F}^\nu\bdot\bV,$ the power consumption from drag, instantaneously in time: 
\begin{eqnarray} 
D^\nu &=& -\rho\int_\Omega \bv_\phi\bdot(\bv^\nu\btimes\bomega^\nu-\nu\grad\btimes\bomega^\nu) dV \cr 
&=& -\rho \int_\Omega \grad\bv_\phi\bdots\bv_\omega^\nu\otimes\bv_\omega^\nu\, dV + 
\rho\int_{\partial\Omega} \bv_\phi\bdot\btau_w^\nu\, dA
\lb{dJA-bd} \end{eqnarray} 
The expression for ${\mathcal T}^\nu$ in the second line is obtained from the first by 
simple vector calculus identities and spatial integration by parts. 

As already remarked in \cite{eyink2021josephson}, the latter form of the Josephson-Anderson 
relation should be valid even for weak Euler solutions $\bv$ obtained as a zero-viscosity limit
of the Navier-Stokes solution $\bv^\nu$. It is the purpose of the present paper to derive relation 
\eqref{dJA-bd} rigorously in the limit $\nu\to 0,$ under reasonable hypotheses which 
are sufficient for a weak Euler solution to exist \cite{drivas2019remarks}. 
This approach is obviously connected with the ``ideal turbulence'' theory of Onsager 
\cite{onsager1949statistical}, who posited that infinite Reynolds-number turbulent flows
will be described by weak Euler solutions dissipating kinetic energy. For reviews, 
see \cite{eyink2018review,delellis2019turbulence}. In particular, we shall follow 
the approach of Duchon \& Robert \cite{duchon2000inertial} to show that the local 
energy balance equations \eqref{Eint-loc}, \eqref{Eom-loc} remain valid in the 
distributional sense for the inviscid limit. The sink term $Q^\nu(\bv^\nu)=
\nu|\bomega^\nu|^2$ which appears in Eq.\eqref{Eom-loc} and which describes the viscous
dissipation in the rotational wake behind the body will be shown to converge
to a non-negative distribution $Q(\bv)$ which corresponds to ``inertial energy dissipation''
or nonlinear energy cascade in the Euler solution. We thus derive a precise 
connection in the inviscid limit between the Josephson-Anderson relation and 
Onsager's dissipative anomaly. 

\subsection{Prior Work} 
Our analysis here will depend essentially upon the approach and results 
in our preceding paper \cite{quan2022inertial}, which studied the conditions 
necessary for a momentum anomaly to exist in wall-bounded turbulence, a possibility 
conjectured already by Taylor \cite{taylor1915eddy}. As remarked at the end of \cite{quan2022inertial},
all of the main results of that work carry over to the local balance of vortex 
momentum given by Eq.\eqref{NS-omega-mom2}. As there, we assume that those equations 
admit strong solutions for arbitrarily large Reynolds numbers.  Then, Theorem 1 of \cite{quan2022inertial} 
implies in this context that distributional limits of the wall shear stress $\btau_w^\nu$
and the rotational pressure stress $p_{\omega,w}^\nu\bn$ at the wall exist for $\nu\to 0:$ 
\begin{align}
    \btau^{\nu}_w &\xrightarrow{\nu\to0} \btau_w \text{ in } D'((\partial B)_T,\mathcal{T}(\partial B)_T) \lb{tau-lim}\\
     p^{\nu}_\omega\mathbf{n} &\xrightarrow{\nu\to0} p_{\omega,w}\mathbf{n} \text{ in } D'((\partial B)_T,\mathcal{N}(\partial B)_T) \lb{p-lim}
\end{align}
More precisely, these limits exist as distributional sections of the tangent and normal bundles, 
respectively, of the space-time manifold $(\partial B)_T=\partial B\times (0,T),$ where we 
assume that $B\subset {\mathbb R}^3$ is closed, bounded, and connected, the complement 
$\Omega={\mathbb R}^3\backslash B$ is also connected, and the common boundary 
$\partial\Omega=\partial B$ is a $C^\infty$ manifold embedded in ${\mathbb R}^3.$
See \cite{quan2022inertial}, section 2 for our notations and conventions regarding distribution theory 
on manifolds. The limit results \eqref{tau-lim},\eqref{p-lim} require hypotheses regarding 
the strong convergence $(\bv_\omega^\nu,p_\omega^\nu)\to (\bv_\omega,p_\omega)$ as $\nu\to 0,$
given precisely in the Theorem \ref{theorem1} below, and the limits are then weak solutions of 
the inviscid version of Eq.\eqref{NS-omega-mom2}:
\begin{eqnarray}
\partial_t\bv_\omega
&=& -\grad\bdot (\bv_\omega\otimes\bv_\omega+\bv_\omega\otimes\bv_\phi+\bv_\phi\otimes\bv_\omega) 
-\grad p_\omega, \quad\grad\bdot\bv_\omega=0. \cr
&& 
\lb{E-omega-mom2} \end{eqnarray} 

The second main result of \cite{quan2022inertial} had to do with the spatial cascade of momentum 
and its matching to stress at the wall. In the present context, the spatial momentum-flux or stress 
consists of the advective contribution: 
\be \mathbf{T}_\omega \bdots=\mathbf{u}_\omega\otimes\mathbf{u}_\omega 
    +\mathbf{u}_\omega\otimes\mathbf{u}_\phi+\mathbf{u}_\phi\otimes\mathbf{u}_\omega \ee 
and the isotropic pressure contribution $p_\omega\bI.$ A turbulent cascade of rotational momentum 
is defined for the infinite-$Re$ Euler solution by introducing fields $\bar{\mathbf{T}}_{\omega,\ell}
=G_\ell*\mathbf{T}_\omega,$ $\bar{p}_{\omega,\ell}=G_\ell*p_\omega$ spatially coarse-grained at length-scale $\ell$
with a standard mollifier $G$ and a smooth window function $\eta_{h,\ell}(\bx)=\theta_{h,\ell}(d(\bx))$
which =0 when distance to the boundary $d(\bx)<h$ for $h>\ell.$ The turbulent flux of momentum 
toward the wall is thus defined by the quantity
\be - (\grad\eta_{h,\ell}\bdot\Bar{\mathbf{T}}_{\omega,\ell} 
    + \Bar{p}_{\omega,\ell}\grad\eta_{h,\ell}). \lb{momflux} \ee 
To justify this physical interpretation it is important to note that for sufficiently small     
$\epsilon<\eta(\Omega),$ for $\Omega_\epsilon:=\{\bx\in \Omega:\, d(\bx)<\epsilon\},$ and
for any $\bx\in\bar{\Omega}_{\epsilon}$, there exists a unique point $\pi(\bx)\in\partial \Omega$ such that
        \begin{align}
            d(\bx) = |\bx - \pi(\bx)|, \;\;\;\; \grad d(\bx) = \bn(\pi(\bx)). 
        \end{align}
Thus, $\grad\eta_{h,\ell}(\bx)=\theta_{h,\ell}'(d(\bx))\bn(\pi(\bx))$ is the product of an approximate 
delta function $\theta_{h,\ell}'(d(\bx))\simeq \delta(d(\bx)-h)$ and the wall normal vector  
$\bn(\bx):=\bn(\pi(\bx))$ extended into the interior, so that \eqref{momflux} indeed measures 
momentum transfer toward the wall at distance $h.$ See section 2 in \cite{quan2022inertial} 
for a detailed discussion of these results. Note also for later use that the unit 
normal vector $\bn$, projection map $\pi$ and distance function $d$ all belong to 
$C^\infty(\bar{\Omega}_\epsilon)$ for each $\epsilon<\eta(\Omega).$  
    
With this background, 
Theorem 2 of \cite{quan2022inertial} states that the wall-parallel component of the flux \eqref{momflux}  
converges in the limit $\ell<h\to 0$ to the wall tangential stress $\btau_w$ and that the wall-normal 
component converges to the wall pressure stress $-p_{\omega,w}\bn.$ For a precise statement, one must 
be able to interpret \eqref{momflux} as a sectional distribution of the tangent bundle and also of the 
normal bundle. For this purpose, \cite{quan2022inertial} introduced a non-standard space of test functions 
\begin{eqnarray}
    \Bar{D}(\Bar{\Omega}\times(0,T)) &:=& \left\{\varphi = \phi|_{\Bar{\Omega}\times(0,T)}:\, 
    \phi\in C_c^{\infty}(\mathbb{R}^3\times(0,T)),\right. \cr 
    && \left. \vspace{3in}\text{supp}(\phi)\cap\left(\Omega\times (0,T)\right)\ne\emptyset\right\}
\end{eqnarray} 
which are non-vanishing on $\partial\Omega$, and also sets ${\mathcal E}_T$ and ${\mathcal E}_N$ 
of {\it extension operators} where $\textbf{Ext}\in {\mathcal E}_T$ is a continuous, linear map  
$\textbf{Ext}: \bpsi\in D((\partial B)_T,\mathcal{T}^*(\partial B)_T) \mapsto 
\bphi \in \Bar{D}(\Bar{\Omega}\times (0,T),\mathbb{R}^3)$ such that the restriction 
$\bphi|_{(\partial B)_T}$ agrees with $\bpsi,$ and likewise 
$\textbf{Ext}\in {\mathcal E}_N$ is a continuous, linear map  
$\textbf{Ext}: \bpsi\in D((\partial B)_T,\mathcal{N}^*(\partial B)_T) \mapsto 
\bphi \in \Bar{D}(\Bar{\Omega}\times (0,T),\mathbb{R}^3)$ such that 
$\bphi|_{(\partial B)_T}=\bpsi.$ The precise statement of 
Theorem 2 in \cite{quan2022inertial} applied to vortex momentum is that for all $\textbf{Ext}\in {\mathcal E}_T$          
    \be   - \lim_{h,\ell\to0}\textbf{Ext}^*(\grad\eta_{h,\ell}\bdot\Bar{\mathbf{T}}_{\omega,\ell} 
    + \Bar{p}_{\omega,\ell}\grad\eta_{h,\ell}) = \btau_w \mbox{ in }D'((\partial B)_T,\mathcal{T}(\partial B)_T)\label{tangentialLimit} \ee
and for all $\textbf{Ext}\in {\mathcal E}_N$ 
    \be
        -\lim_{h,\ell\to0} \textbf{Ext}^*(\grad\eta_{h,\ell}\bdot\Bar{\mathbf{T}}_{\omega,\ell} 
        + \Bar{p}_{\omega,\ell}\grad\eta_{h,\ell}) = -p_{\omega,w}\mathbf{n} \mbox{ in }D'((\partial B)_T,\mathcal{N}(\partial B)_T). \label{normalLimit} \ee
        
The final main result of \cite{quan2022inertial} was that, under some additional tenable 
hypotheses, the advective flux of vortex momentum to the wall must vanish. The added 
assumptions were, for some $\epsilon>0,$ the near-wall boundedness property 
            \begin{align}
                \mathbf{u}_\omega \in L^2((0,T),L^{\infty}(\Omega_{\epsilon}))
            \end{align}
with $\Omega_\epsilon:=\{\bx\in \Omega:\, {\rm dist}(\bx,\partial B)<\epsilon\}$ and the no-flow-through condition 
      \begin{align}
            \lim_{\delta\to0}\norm{\mathbf{n}\bdot\mathbf{u}_\omega}_{L^2((0,T),L^{\infty}(\Omega_{\delta}))} = 0.
        \end{align}
Then Theorem 3 of \cite{quan2022inertial} implies that,  for all $\textbf{Ext}\in \Tilde{\cE}_{\cT},$
\be
-\lim_{h,\ell\to 0} \textbf{Ext}^*(\grad\eta_{h,\ell}\bdot\Bar{\mathbf{T}}_{\omega,\ell}) 
\,= \btau_w\,= \bzed \mbox{ in } D'((\partial B)_T, \mathcal{T}(\partial B)_T)
\lb{tauzed} 
\ee 
and for all $\textbf{Ext}\in {\mathcal E}_{\cN},$
\be 
-\lim_{h,\ell\to 0} \textbf{Ext}^*(\Bar{p}_{\omega,\ell}\grad\eta_{h,\ell}) = -p_{\omega,w}\mathbf{n}, \mbox{ in } D'((\partial B)_T, \mathcal{N}(\partial B)_T). 
\lb{pcont} 
\ee 
Here $\tilde{\mathcal{E}}_{\mathcal{T}}$ is a class of {\it natural extensions} which consists of those 
$\textbf{Ext}\in\mathcal{E}_{\mathcal{T}}$ such that $\forall\bpsi\in D((\partial B)_T,\mathcal{T}^*(\partial B)_T)$, $\bphi = \textbf{Ext}(\bpsi)$ 
satisfies 
   \begin{align}
        \norm{\bphi\bdot\mathbf{n}}_{L^{\infty}((\Omega_{h+\ell}\backslash\Omega_h)\times(0,T))}=O(\ell),
        \quad h,\ell\to0,\label{extCond1}
    \end{align}
which allows to infer that $\Bar{p}_{\omega,\ell}\grad\eta_{h,\ell}$ gives vanishing contribution to the 
tangential wall stress in \eqref{tangentialLimit}, as would be naturally expected. 
It follows then from \eqref{tauzed} that there is no viscous momentum anomaly, $\btau_w=\bzed,$ so that all drag originates
asymptotically from pressure forces and, furthermore, from \eqref{pcont} the pressure of the Euler solution away from the wall 
matches onto the pressure at the wall. 

The present work will build on these these results for vortex momentum and extend them to rotational kinetic energy 
and to interaction energy of potential and rotational flow, obtaining the infinite Reynolds-number limit of the 
Josephson-Anderson relation in the process. Because kinetic energies are scalar quantities, we can simplify 
somewhat the approach of \cite{quan2022inertial}. Following \cite{wagner2010distributions}, we denote by $D((\partial B)_T)$ 
the $C^\infty$ sections with compact support of the trivial scalar bundle $(\partial B)_T\times {\mathbb R}$ and by 
$D'((\partial B)_T)$ the sectional distributions of that bundle. We refer to these simply as spaces of scalar test functions 
and of scalar distributions on $(\partial B)_T,$ respectively. We then define the class of scalar extensions  
${\mathcal E}$ as the set of continuous, linear maps   
$\textbf{Ext}: \psi\in D((\partial B)_T) \mapsto 
\phi \in \Bar{D}(\Bar{\Omega}\times (0,T))$ such that the restriction 
$\phi|_{(\partial B)_T}=\psi.$ This set ${\mathcal E}$ is defined more precisely 
in the following section \ref{prelim}, where it is shown also that ${\mathcal E}\neq\emptyset$ by 
construction of a concrete example. We shall need further in our proof the fact that the stationary Euler solution 
describing potential flow around the body $B$ with asymptotic velocity $\bV$ at infinity
satisfies $\bv_\phi\in C^\infty(\bar{\Omega}),$
which follows from the known smoothness of the solution of the Laplace equation with zero Neumann 
conditions on a domain $\Omega$ with smooth boundary $\partial \Omega.$  We shall review these
results also in section \ref{prelim}. Since $\bv_\phi\bdot\bn=0$ on $\partial B$ and since 
$\partial B$ is compact, it follows that we may interpret $\bv_\phi|_{\partial B}\in D((\partial B)_T,{\mathcal T}^*(\partial B)_T).$
Thus, the dot product with the distribution $\btau_w\in D'((\partial B)_T,{\mathcal T}(\partial B)_T)$ 
obtained by Theorem 1 of \cite{quan2022inertial} can be defined with $\bv_\phi\bdot\btau_w\in D'((\partial B)_T)$ 
by setting
\be \langle \bv_\phi\bdot\btau_w,\psi\rangle:=\langle\btau_w,\psi\bv_\phi|_{\partial B}\rangle, 
\quad \forall \psi \in D((\partial B)_T). \ee

\subsection{Results on Interaction Energy and Josephson-Anderson Relation}

We can now state the first main result of the present work, which is analogous to Theorem 1 
of \cite{quan2022inertial} for vortex momentum, but here on the infinite Reynolds-number limits 
of the local balance of the interaction energy \eqref{Eint-loc} and of the Josephson-Anderson 
relation \eqref{dJA-bd}: 

\begin{theorem}\label{theorem1} 
We make the following assumptions (which are same as those in Theorem 1 of \cite{quan2022inertial}):
    Let $(\mathbf{u}^{\nu}_\omega, p^{\nu}_\omega)$ be strong solutions of Eq.\eqref{NS-omega-mom2} on $\Bar{\Omega}\times (0,T)$ for $\nu>0$. 
    Assume that
    $(\mathbf{u}^{\nu}_{\omega})_{\nu>0}$ strongly converges to $\mathbf{u}_\omega$ in $L^2((0,T),L_{\text{loc}}^2(\Omega)):$
    %, and initial data $(\mathbf{u}^{\nu}_0)_{\nu>0}$ weakly converges to $\mathbf{u}_0$ in $L_{\text{loc}}^2(\Omega)$ i.e.
        %\be
        %    \mathbf{u}^{\nu}\xrightarrow[L^2(0,T;L_{\text{loc}}^2(\Omega))]{\nu\to0}\mathbf{u}, \;\; %\mathbf{u}^{\nu}_0\xrightarrow[L_{\text{loc}}^2(\Omega)]{\nu\to0}\mathbf{u}_0\label{L2Conv}
        %\ee
        \be
            \mathbf{u}^{\nu}_{\omega}\xrightarrow[L^2((0,T),L_{\text{loc}}^2(\Omega))]{\nu\to0}\mathbf{u}_\omega. \label{L2Conv}
        \ee    
    and that $(p^{\nu}_{\omega})_{\nu>0}$ strongly converges to $p_\omega$ in $L^1((0,T),L_{\text{loc}}^1(\Omega)):$    
    \be
            p^{\nu}_{\omega}\xrightarrow[L^1((0,T),L_{\text{loc}}^1(\Omega))]{\nu\to0}p_\omega. \label{pL1Conv}
        \ee    
    Assume furthermore for some $\epsilon>0$ arbitrarily small%, with $\Omega_\epsilon:=\{\bx\in \Omega:\, dist(\bx,\partial B)<\epsilon\},$
        \begin{align}
            &\mathbf{u}^{\nu}_\omega \text{ uniformly bounded in } L^2((0,T),L^2(\Omega_{\epsilon}))\label{uBBound}\\
            %&\mathbf{u}^{\nu}_0 \text{ uniformly bounded in } L^2(\Omega_{\epsilon})\label{u0BBound}\\
            &p^{\nu}_\omega \text{ uniformly bounded in } L^1((0,T), L^1(\Omega_{\epsilon})).\label{pBBound}
        \end{align}
    \iffalse    
    In addition, we assume that $Q^{\nu} = \nu|\bomega^{\nu}|^2$ converges to a Radon measure $Q$ on $\Bar{\Omega}\times(0,T)$ 
    in the sense that $\forall\varphi\in \Bar{D}(\Bar{\Omega}\times (0,T))$,
        \begin{align}
            \lim_{\nu \to 0}  \int_0^T \int_{\Omega} \varphi\,Q^\nu\,  dV\, dt = \int_0^T \int_{\Omega} 
            \varphi\, Q \ dV\, dt\label{viscDissLimit}
        \end{align}
    \fi
    Then, the limit $(\bv_\omega,p_\omega)$ is a weak solution of \eqref{E-omega-mom2} and for this same sequence  
    and for all $\mathbf{Ext}\in {\mathcal E},$
   \begin{align}
         \lim_{\nu\to 0} \mathbf{Ext}^*[\bv_\phi\bdot\nu\grad\btimes\bomega^\nu]
        = \bv_\phi\bdot \btau_w 
        \quad \mbox{ in $D'((\partial B)_T)$}
    \label{uphi-tau} \end{align}
    Furthermore, the local balance Eq.\eqref{Eint-loc} for the interaction energy
    holds distributionally in the sense that for all $\varphi \in \Bar{D}(\Bar{\Omega}\times (0,T))$, with $\psi = \varphi|_{\partial B}$, 
    \begin{eqnarray} 
        %- \int_{\Omega}\varphi(\mathbf{x},0)(\mathbf{u}_{\omega}\cdot\mathbf{u}_{\phi})(\mathbf{x},0) d\mathbf{x} 
        && -\int_0^T\int_{\Omega}\partial_t\varphi\,(\mathbf{u}_{\omega}\bdot\mathbf{u}_{\phi}) \, dV\,dt\cr 
        && -\int_0^T\int_{\Omega}\grad\varphi\bdot\left[(\mathbf{u}_{\omega}\bdot\mathbf{u}_{\phi})\mathbf{u}+\frac{1}{2}|\mathbf{u}_{\phi}|^2\mathbf{u}_{\omega}+p_{\omega}\mathbf{u}_{\phi}+p_{\phi}\mathbf{u}_{\omega}\right] dV\,dt\cr 
        && =-\langle\bv_\phi\bdot\btau_w,\psi\rangle
        +\int_0^T\int_{\Omega}\varphi\,(\grad\mathbf{u}_{\phi}\bdots\mathbf{u}_{\omega}\otimes\mathbf{u}_{\omega})\,dV\,dt
        \label{weakInteractionLimit1}
    \end{eqnarray}
    Finally, if \eqref{L2Conv} is strengthened to global $L^2$ convergence
    \be
            \mathbf{u}^{\nu}_{\omega}\xrightarrow[L^2((0,T),L^2(\Omega))]{\nu\to0}\mathbf{u}_\omega, \label{L2Conv-str}
    \ee    
    then the Josephson-Anderson relation \eqref{dJA-bd} holds also in the inviscid limit, in the sense that
    $D=\lim_{\nu\to 0} {\bf F}^\nu\bdot\bV$ exists in ${\mathcal D}'((0,T))$ with 
    \be
\langle D,\chi\rangle =
-\rho \int_\Omega \chi \grad\bv_\phi\bdots\bv_\omega\otimes\bv_\omega\, dV\,dt + 
\rho\langle \bv_\phi\bdot\btau_w,\chi\otimes 1\rangle 
\lb{dJA-zero} \ee 
for all $\chi\in {\mathcal D}((0,T)).$
\end{theorem}

%\begin{remark}
%\be \bv\stackrel{{\!\,}_\circ}{\btimes}\bomega:=-\grad\bdot(\bv\bv) +\grad(\frac{1}{2}|\bv|^2) \ee
%\end{remark}

\begin{remark}
Although we have so far presented our results by assuming a fluid velocity at infinity $\bV$ which is 
time-independent, the proof presented below should allow any function $\bV(t)$ which is $C^\infty$ in time, 
corresponding to smooth translational motion of the body $B.$ In that case the potential Euler solution $\bv_\phi$
is of course also time-dependent. See footnote [91] in \cite{eyink2021josephson} and section \ref{sec:pot} below for more
discussion. It would be interesting to generalize these results further to allow for solid-body dynamics of $B$ 
including possible rotation, as in \cite{sueur2012kato}. 
\end{remark} 

\begin{remark} Just as in \cite{quan2022inertial}, Theorem 1, condition \eqref{pL1Conv} on convergence of pressure 
can be replaced with any assumption guaranteeing that along a suitable subsequence of $\nu$, $p^{\nu}_\omega\to p_\omega
    \in  L^1((0,T), L^1_{loc}(\Omega))$ in the sense of distributions. E.g. \\
    
    \vspace{-20pt} 
        \begin{align}
            &p^{\nu}_\omega \text{ uniformly bounded in }L^{q}((0,T),L_{\text{loc}}^{q}(\Omega))
            \mbox{ for some $q>1$} 
            \label{pBulkBound}
        \end{align}
would suffice. See \cite{quan2022inertial}, Remark 2.         
\end{remark}

\iffalse
\begin{remark}
Formally, the force is ${\bf F}= \rho\int_{\partial B} (-p_\omega\bn+\btau_w) dA$ in the inviscid 
limit or, precisely, for $\chi\in D((0,T))$ by 
$\langle F_i,\chi\rangle = \rho[-\langle p_\omega\bn,n_i\chi\bn\rangle+ \sum_{\alpha=1,2}\langle\btau_w,e_{\alpha i}\chi{\bf e}_\alpha\rangle]$ 
where ${\bf e}_\alpha\in D(\partial B, \mathcal{T}^*(\partial B)),$ $\alpha=1,2$ are an orthogonal pair of smooth unit cotangent 
vector fields on $\partial B,$ and $e_{\alpha,i}$ and $n_i$ are the Cartesian components of ${\bf e}_\alpha$ and 
of the normal vector field $\bn$ in ${\mathbb R}^3$ through the embeddings $\biota_T,\biota_N.$ For each $i=1,2,3,$
we then have $n_i\chi\bn\in D((\partial B)_T,\mathcal{N}^*(\partial B)_T),$  
$e_{\alpha i}\chi{\bf e}_\alpha\in D((\partial B)_T,\mathcal{T}^*(\partial B)_T),$ $\alpha=1,2,$
so that the stated expression for $\langle F_i,\chi\rangle$ then follows by Theorem 1 of \cite{quan2022inertial}. 
KEEP????
\end{remark}
\fi

Our next result is analogous to Theorem 2 in \cite{quan2022inertial}, which involved
spatial momentum cascade, but considering now the spatial flux of interaction energy. 
Similar to the definition \eqref{momflux} of momentum flux, we apply spatial coarse-graining
and windowing to Eq.\eqref{NS-omega-mom2} for $\bv_\omega^\nu$ and to the Euler equation for $\bv_\phi$, 
yielding the balance for interaction energy in length scales $>\ell$ and wall distances $>h:$
\begin{align}
    &\partial_t(\eta_{h,\ell}\Bar{\mathbf{u}}_{\omega,\ell}^{\nu}\bdot\Bar{\mathbf{u}}_{\phi,\ell}) + \grad\bdot(\eta_{h,\ell}\Bar{\mathbf{J}}_{\phi,\ell}^{\nu})\label{cgInteractionWindow1}
    =\grad\eta_{h,\ell}\bdot\Bar{\mathbf{J}}_{\phi,\ell}^{\nu} + \eta_{h,\ell}\Bar{\mathbf{u}}_{\omega,\ell}^{\nu}\Bar{\mathbf{u}}_{\omega,\ell}^{\nu}\bdots\grad\Bar{\mathbf{u}}_{\phi,\ell}\\ \nonumber
    &+\eta_{h,\ell}\grad\Bar{\mathbf{u}}_{\phi,\ell}\bdots(\tau_\ell(\bv_\omega^\nu,\bv_\omega^\nu) 
    +\tau_\ell(\bv_\phi,\bv_\omega^\nu)+\tau_\ell(\bv_\omega^\nu,\bv_\phi))
    -\eta_{h,\ell}\Bar{\mathbf{u}}_{\omega,\ell}^{\nu}\bdot(\grad\bdot\tau_\ell(\bv_\phi,\bv_\phi)). 
\end{align}
Here $\tau_\ell(f,g):=\overline{(fg)}_\ell-\bar{f}_\ell\bar{g}_\ell$ for any functions $f,$ $g$ on 
$\Omega^\ell:=\{\bx\in \Omega: d(\bx)>\ell\}$ and spatial flux of interaction energy is given by 
\begin{eqnarray} 
    \Bar{\mathbf{J}}_{\phi,\ell}^{\nu} = (\Bar{\mathbf{u}}_{\omega,\ell}^{\nu}\bdot\Bar{\mathbf{u}}_{\phi,\ell})\Bar{\mathbf{u}}^{\nu}_\ell
    && \! +\ \Bar{p}_{\omega,\ell}^{\nu}\Bar{\mathbf{u}}_{\phi,\ell}
    +\Big(\frac{1}{2}|\Bar{\mathbf{u}}_{\phi,\ell}|^2+\Bar{p}_{\phi,\ell}\Big)\Bar{\mathbf{u}}_{\omega,\ell}^{\nu}
    -\nu\Bar{\mathbf{u}}_{\phi,\ell}\btimes\Bar{\bomega}^{\nu}_\ell \cr 
    && \! +\ (\tau_\ell(\bv_\omega^\nu,\bv_\omega^\nu) 
    +\tau_\ell(\bv_\phi,\bv_\omega^\nu)+\tau_\ell(\bv_\omega^\nu,\bv_\phi))\bdot\Bar{\mathbf{u}}_{\phi,\ell} 
\label{Jphi-def} \end{eqnarray} 
The balance equation \eqref{cgInteractionWindow1} holds also for the limiting fields $(\bv_\omega,p_\omega)$ 
of Theorem \ref{theorem1}, as stated by the following Proposition, whose easy proof is left to 
the reader:
\begin{proposition}
Under the assumptions \eqref{L2Conv}-\eqref{pL1Conv} of Theorem \ref{theorem1}, the following equation 
holds pointwise for $\bx\in\Omega$ and distributionally for $t\in (0,T)$:
\begin{eqnarray} 
    &&\partial_t(\eta_{h,\ell}\Bar{\mathbf{u}}_{\omega,\ell}\bdot\Bar{\mathbf{u}}_{\phi,\ell}) 
    + \grad\bdot(\eta_{h,\ell}\Bar{\mathbf{J}}_{\phi,\ell})
    =\grad\eta_{h,\ell}\bdot\Bar{\mathbf{J}}_{\phi,\ell} + \eta_{h,\ell}\Bar{\mathbf{u}}_{\omega,\ell}\Bar{\mathbf{u}}_{\omega,\ell}\bdots\grad\Bar{\mathbf{u}}_{\phi,\ell}
    \cr
    &&+\eta_{h,\ell}\grad\Bar{\mathbf{u}}_{\phi,\ell}\bdots(\tau_\ell(\bv_\omega,\bv_\omega) 
    +\tau_\ell(\bv_\phi,\bv_\omega)+\tau_\ell(\bv_\omega,\bv_\phi))
    -\eta_{h,\ell}\Bar{\mathbf{u}}_{\omega,\ell}\bdot(\grad\bdot\tau_\ell(\bv_\phi,\bv_\phi)). \cr 
    &&  \label{cgWindowEulerInteraction}
\end{eqnarray} 
where $\Bar{\mathbf{J}}_{\phi,\ell}$ is given by the formula \eqref{Jphi-def} with $\nu\mapsto 0,$  
$(\bv_\omega^\nu,p_\omega^\nu)\mapsto (\bv_\omega,p_\omega).$ 
    \end{proposition}
\noindent Our second main result is then:
\begin{theorem}\label{theorem2}
    Assume conditions (\ref{L2Conv})-(\ref{pBBound}) in Theorem \ref{theorem1}. Then 
%    \eqref{cgwInteractionWindow2}
%    %(\ref{cgEulerInteractionEq1})-(\ref{cgEulerInteractionEq2}) 
%    converges as $h,l\to0$, to the following limiting distributional equation of interaction energy of inviscid flow:
%    \begin{align}
%        &- \int_{\Omega}\varphi(\mathbf{x},0)(\mathbf{u}_{\omega}\cdot\mathbf{u}_{\phi})(\mathbf{x},0) d\mathbf{x} %-\int_0^T\int_{\Omega}\partial_t\varphi(\mathbf{u}_{\omega}\cdot\mathbf{u}_{\phi})d\mathbf{x}dt\\
%        &-\int_0^T\int_{\Omega}\grad\varphi\cdot\left[(\mathbf{u}_{\omega}\cdot\mathbf{u}_{\phi})\mathbf{u}+\frac{1}{2}|\mathbf{u}_{\phi}|^2\mathbf{u}_{%\omega}+p_{\omega}\mathbf{u}_{\phi}+p_{\phi}\mathbf{u}_{\omega}\right]d\mathbf{x}dt\\
%        &= \lim_{h,l\to0}\int_0^T\int_{\Omega}\varphi\grad\eta_{h,l}\cdot\Bar{\mathbf{J}}_{\phi}^{k,0}d\mathbf{x}dt 
%        +\int_0^T\int_{\Omega}\varphi(\grad\mathbf{u}_{\phi}:\mathbf{u}_{\omega}\mathbf{u}_{\omega})d\mathbf{x}dt
%    \end{align}
%    and by comparison with (\ref{weakInteractionLimit1})-(\ref{weakInteractionLimit2}), we have
\begin{align}
    - \lim_{h,\ell\to0} \mathbf{Ext}^*[\grad\eta_{h,\ell}\bdot\Bar{\mathbf{J}}_{\phi,\ell}] = \bv_\phi\bdot \btau_w 
    \quad \mbox{ in $D'((\partial B)_T)$} \label{cascade_friction}
\end{align}
for all $\mathbf{Ext}\in {\mathcal E}.$ 
\end{theorem}
\noindent This theorem states that the spatial cascade of interaction energy to the wall equals the 
skin friction term in the inviscid Josephson-Anderson relation \eqref{dJA-zero}.  

Finally, Theorem 3 of \cite{quan2022inertial} showed that momentum cascade to the wall must 
vanish when the limiting Euler solution satisfies the no-flow-through condition at the wall 
in a suitable sense. An exact analogue of this result for interaction energy is given by the following: 

\begin{theorem}\label{theorem3} 
Let $(\mathbf{u}_\omega, p_\omega)$ be the weak solution of Eq.\eqref{E-omega-mom2} from Theorem \ref{theorem1}, with 
$\mathbf{u}_\omega\in L^2(0,T,L^2_{\text{loc}}(\Omega))\cap L^2(0,T,L^2(\Omega_{\epsilon})),$  
$p_\omega\in L^1(0,T,L_{\text{loc}}^1(\Omega))\cap L^1(0,T,L^1(\Omega_{\epsilon})$ for some $\epsilon>0.$ 
Assume further the stronger near-wall boundedness property 
            \begin{align}
                \mathbf{u}_\omega \in L^2((0,T),L^{\infty}(\Omega_{\epsilon})), \;\;\;\; p_{\omega}\in L^1((0,T),L^{\infty}(\Omega_{\epsilon}))\label{wallboundedness}
            \end{align}
and the no-flow-through condition at the boundary in the sense
      \begin{align}
            \lim_{\delta\to0}\norm{\mathbf{n}\bdot\mathbf{u}_\omega}_{L^2((0,T),L^{\infty}(\Omega_{\delta}))} = 0. \label{wallnormal}
        \end{align}
     Then, for all $\mathbf{Ext}\in {\mathcal E}$ 
     \begin{align} 
        \hspace{50pt} \lim_{h,\ell\to0} \mathbf{Ext}^*[\grad\eta_{h,\ell}\bdot\Bar{\mathbf{J}}_{\phi,\ell}] = 0 
        \quad \mbox{ in $D'((\partial B)_T)$}
    \end{align}
so that $\bv_\phi\bdot \btau_w=0.$ Thus, an inviscid form of the balance Eq.\eqref{Eint-loc} for the interaction energy holds distributionally, in the sense that for all $\varphi \in \Bar{D}(\Bar{\Omega}\times (0,T))$
    \begin{align}
        %- \int_{\Omega}\varphi(\mathbf{x},0)(\mathbf{u}_{\omega}\cdot\mathbf{u}_{\phi})(\mathbf{x},0) d\mathbf{x} 
        & -\int_0^T\int_{\Omega}\partial_t\varphi\,(\mathbf{u}_{\omega}\bdot\mathbf{u}_{\phi}) \, dV\,dt\label{weakInteractionLimit2} 
        \\ \nonumber
        & -\int_0^T\int_{\Omega}\grad\varphi\bdot\left[(\mathbf{u}_{\omega}\bdot\mathbf{u}_{\phi})\mathbf{u}+\frac{1}{2}|\mathbf{u}_{\phi}|^2\mathbf{u}_{\omega}+p_{\omega}\mathbf{u}_{\phi}+p_{\phi}\mathbf{u}_{\omega}\right] dV\,dt\\ \nonumber 
        &\hspace{80pt} 
        =\int_0^T\int_{\Omega}\varphi\,(\grad\mathbf{u}_{\phi}\bdots\mathbf{u}_{\omega}\otimes\mathbf{u}_{\omega})\,dV\,dt.
    \end{align}
    If convergence to the inviscid limit holds in the global $L^2$-sense \eqref{L2Conv-str}, then the Josephson-Anderson relation \eqref{dJA-bd} is valid in the inviscid form 
    \be
D =
-\rho \int_\Omega \grad\bv_\phi\bdots\bv_\omega\otimes\bv_\omega\, dV.
\lb{dJA-zero2} \ee 
\end{theorem}
\noindent The implication is that cascade of interaction energy to the wall must vanish if the strong 
near-wall boundedness and no-flow-through conditions are satisfied.  

\subsection{Results on Energy Balance of the Rotational Flow} We next wish to show that  
results hold analogous to those of Duchon \& Robert on turbulent kinetic 
energy balance \cite{duchon2000inertial}, but now for the energy in the rotational wake. We have: 

\begin{theorem}\label{theorem4} 
    Let $(\mathbf{u}^{\nu}_\omega, p^{\nu}_\omega)$ be strong solutions of equations \eqref{NS-omega-mom2} 
    on $\Bar{\Omega}\times(0,T)$. Assume that
    $(\mathbf{u}^{\nu}_\omega)_{\nu>0}$ strongly converges to $\mathbf{u}_\omega$ in $L^3((0,T),L_{\text{loc}}^3(\Omega))$, 
    %and initial data $(\mathbf{u}^{\nu}(0))_{\nu>0}$ weakly converges to $\mathbf{u}_0$ in $L_{\text{loc}}^2(\Omega)$ i.e.
        \begin{align}
            \mathbf{u}^{\nu}_\omega\xrightarrow[L^3((0,T),L_{\text{loc}}^3(\Omega))]{\nu\to0}\mathbf{u}_\omega, 
            %\;\; \mathbf{u}^{\nu}\xrightarrow[L_{\text{loc}}^2(\Omega)]{\nu\to0}\mathbf{u}_0
            \label{L3Conv}
            \end{align}
    and         
            \begin{align} 
            p_{\omega}^{\nu} \text{ uniformly bounded in } L^{\frac{3}{2}}((0,T),L^{\frac{3}{2}}_{loc}(\Omega)). 
            \lb{pL32Bd} 
        \end{align}
    Assume further for some $\epsilon>0$ 
        \begin{align}
            &\mathbf{u}^{\nu}_\omega \text{ uniformly bounded in } L^3((0,T),L^3(\Omega_{\epsilon}))\label{uL3Boundary}\\
        %    &\mathbf{u}^{\nu}(0) \text{ uniformly bounded in } L^2(\Omega_{\epsilon_0})\label{u0L2Boundary}\\
            &p^{\nu}_\omega \text{ uniformly bounded in } L^\frac{3}{2}((0,T), L^{\frac{3}{2}}(\Omega_{\epsilon})).\label{pL32Boundary}
        \end{align}

    Then, along a suitable subsequence $p^{\nu}_\omega\to p_\omega\in L^{\frac{3}{2}}((0,T),L^{\frac{3}{2}}_{\text{loc}}(\Omega))$ distributionally,  
    so that $(\bv_\omega,p_\omega)$ is a weak solution of the inviscid equation \eqref{E-omega-mom2}
    satisfying the results \eqref{uphi-tau},\eqref{weakInteractionLimit1} of Theorem \ref{theorem1}.  
    Also, $Q^{\nu} = \nu|\bomega^{\nu}|^2$ converges for this subsequence to a positive linear functional $Q$ 
    on $\Bar{D}(\Bar{\Omega}\times (0,T))$, in the sense that $\forall\varphi\in \Bar{D}(\Bar{\Omega}\times (0,T))$,
        \begin{align}
            \lim_{\nu \to 0}  \int_0^T \int_{\Omega} \varphi\, Q^\nu \,  dV\, dt = \langle Q,\varphi\rangle
            \label{viscDissLimit2}
        \end{align}
    with $\langle Q,\varphi\rangle\geq 0$ for $\varphi\geq 0.$ Finally, 
    %we have: for $\varphi \in \Bar{D}(\Bar{\Omega}\times[0,T))$, $\psi = \varphi|_{\partial B}$,
    %\begin{align}
    %    \lim_{\nu\to0}\int_0^T\int_{\Omega}\varphi\grad\cdot(\nu\mathbf{u}_{\omega}^{\nu}\times\bm{\omega})d\mathbf{x}dt
    %    &=\bm{\tau}^*(\psi\mathbf{u}_{\phi}|_{\partial B})
    %\end{align}
   an inviscid version of the balance equation \eqref{Eom-loc} for rotational energy holds in the sense that
   for all $\varphi \in \Bar{D}(\Bar{\Omega}\times(0,T))$, $\psi = \varphi|_{\partial B}$,
    \begin{eqnarray} 
        && %-\int_{\Omega}\frac{1}{2}\varphi|\mathbf{u}_{\omega}|^2(0,x)d\mathbf{x}
        -\int_0^T\int_{\Omega}\frac{1}{2}\partial_t\varphi|\mathbf{u}_{\omega}|^2\,dV\, dt
        -\int_0^T\int_{\Omega}\grad\varphi\bdot\left[\frac{1}{2}|\mathbf{u}_{\omega}|^2\mathbf{u}+p_{\omega}\mathbf{u}_{\omega}\right]
        dV\,dt\cr 
        && \hspace{30pt} =\langle\bv_\phi\bdot \btau_w,\psi\rangle
        -\langle Q,\varphi\rangle  -\int_0^T\int_{\Omega}\varphi\grad\mathbf{u}_{\phi}\bdots
        \mathbf{u}_{\omega}\otimes\mathbf{u}_{\omega}\,dV\,dt
        \label{fgRotationalLimit2}
    \end{eqnarray} 
\end{theorem}

\begin{remark}
The result \eqref{viscDissLimit2} is analogous to Proposition 4 of Duchon-Robert \cite{duchon2000inertial}, 
which showed that the inviscid limit of viscous dissipation exists as a non-negative distribution. 
In fact, the linear functional $Q$ in our theorem restricted to the subspace $D(\Omega\times (0,T))$ 
is a non-negative distribution. Thus, by standard representation theory for distributions 
(\cite{choquet1969lectures}, Example 12.5; \cite{rudin1991functional}, Exercise 6.4),  
there is a non-negative Radon measure on $\Omega\times (0,T),$ which we shall also write as $Q$
by an abuse of notation, so that
\be \langle Q,\varphi\rangle= \int_{\Omega\times (0,T)} \varphi(\bx,t) \, Q(d\bx\,dt),
\quad \forall\varphi\in D(\Omega\times (0,T)). \lb{Qrep} 
\ee
We have not shown that a Radon measure $Q$ exists on $\Bar{\Omega}\times (0,T)$ such that 
the representation \eqref{Qrep} holds for all $\varphi \in \Bar{D}(\Bar{\Omega}\times(0,T)),$
but it is plausible to conjecture that this is true. If so, another natural question 
is whether $Q((\partial B)_T)>0$ or $=0.$ The latter question is related to the famous theorem 
of Kato \cite{kato1984remarks,sueur2012kato}, who showed that, if a smooth Euler solution $\bv$ 
exists over the time interval $(0,T),$ then $\bv^\nu\to\bu$ in $L^2((0,T),L^2(\Omega))$ 
when $\bv^\nu_0\to\bu_0$ in $L^2(\Omega)$ and $Q^\nu(\Omega_{c\nu}\times (0,T))\to 0$ for some constant $c>0.$
Under these specific assumptions, a dissipative Euler solution obtained in the inviscid limit plausibly 
must have $Q((\partial B)_T)>0.$
\end{remark}

\begin{remark}
Under the stronger assumptions of Theorem \ref{theorem4} then also the conclusions 
of Theorem \ref{theorem1} hold, so that one can combine Eq.\eqref{weakInteractionLimit1} of that 
theorem with the present result Eq.\eqref{fgRotationalLimit2} to obtain the combined balance: 
\begin{eqnarray} 
       %- \int_{\Omega}\varphi(\mathbf{x},0)(\mathbf{u}_{\omega}\cdot\mathbf{u}_{\phi})(\mathbf{x},0) d\mathbf{x} 
        && -\int_0^T\int_{\Omega}\partial_t\varphi\,
        \left[ (\mathbf{u}_{\omega}\bdot\mathbf{u}_{\phi})
         + \frac{1}{2}|\mathbf{u}_{\omega}|^2\right] \, dV\,dt\cr 
        && -\int_0^T\int_{\Omega}\grad\varphi\bdot\left[
        \left( (\mathbf{u}_{\omega}\bdot\mathbf{u}_{\phi})
         + \frac{1}{2}|\mathbf{u}_{\omega}|^2+p_\omega\right)\bv
        +\left( \frac{1}{2}|\mathbf{u}_{\phi}|^2+p_\phi\right)\bv_\omega\right]\, dV\,dt\cr 
        && \hspace{120pt} = -\langle Q,\varphi \rangle
        \label{weaktotalLimit2}
\end{eqnarray}      
This combined quantity was termed the ``relative kinetic energy'' in \cite{eyink2021josephson}, because 
it measures the total energy of the fluid relative to the energy of the potential Euler solution
and it is thus locally conserved when $Q=0.$
\end{remark}

Our next result is a close analogue of Theorem \ref{theorem2} but considering now the spatial flux of 
rotational flow energy. By applying spatial coarse-graining and windowing to Eq.\eqref{NS-omega-mom2} 
for $\bv_\omega^\nu,$ we obtain a balance equation for 
rotational energy in length scales $>\ell$ and wall distances $>h:$
\begin{eqnarray} 
    && \partial_t(\frac{1}{2}\eta_{h,\ell}|\Bar{\mathbf{u}}_{\omega,\ell}^{\nu}|^2) + \grad\bdot(\eta_{h,\ell}\Bar{\mathbf{J}}_{\omega,\ell}^\nu) 
    = \grad\eta_{h,\ell}\bdot\Bar{\mathbf{J}}_{\omega,\ell}^\nu
    -\eta_{h,\ell}\grad\Bar{\mathbf{u}}_{\phi,\ell}\bdots\Bar{\mathbf{u}}_{\omega,\ell}^{\nu}
    \Bar{\mathbf{u}}_{\omega,\ell}^{\nu} \cr 
    && \hspace{15pt} +\eta_{h,\ell}\grad\Bar{\mathbf{u}}_{\omega,\ell}^{\nu}
    \bdots(\tau_\ell(\bv_\omega^\nu,\bv_\omega^\nu) +\tau_\ell(\bv_\phi,\bv_\omega^\nu)+\tau_\ell(\bv_\omega^\nu,\bv_\phi))
    - \nu\eta_{h,\ell}|\Bar{\bomega}^{\nu}_\ell|^2\label{cgwRotaionalEq2}
\end{eqnarray}
with spatial flux of rotational energy given by 
\begin{eqnarray} 
&&    \Bar{\mathbf{J}}_{\omega,\ell}^{\nu} = \frac{1}{2}|\Bar{\mathbf{u}}_{\omega,\ell}^{\nu}|^2\Bar{\mathbf{u}} ^{\nu}_\ell
+ \Bar{p}_{\omega,\ell}\Bar{\mathbf{u}}_{\omega,\ell}^{\nu}-\nu\Bar{\mathbf{u}}_{\omega,\ell}^{\nu}\btimes\Bar{\bomega}^{\nu}_\ell \cr
&& \hspace{60pt} + (\tau_\ell(\bv_\omega^\nu,\bv_\omega^\nu) +\tau_\ell(\bv_\phi,\bv_\omega^\nu)+\tau_\ell(\bv_\omega^\nu,\bv_\phi))\bdot\Bar{\mathbf{u}}_{\omega,\ell}^{\nu}
\label{Jom-def} \end{eqnarray} 
The balance equation \eqref{cgwRotaionalEq2} holds also for the limiting fields $(\bv_\omega,p_\omega)$ 
of Theorem \ref{theorem4}, as stated by the following easy Proposition, with proof left to the reader:
\begin{proposition}
Under the assumptions (\ref{L3Conv})-(\ref{viscDissLimit2}) of Theorem \ref{theorem4}, the following equation 
holds pointwise for $\bx\in\Omega$ and distributionally for $t\in (0,T)$:
\begin{eqnarray} 
    && \partial_t(\frac{1}{2}\eta_{h,\ell}|\Bar{\mathbf{u}}_{\omega,\ell}|^2) + \grad\bdot(\eta_{h,\ell}\Bar{\mathbf{J}}_{\omega,\ell}) 
    = \grad\eta_{h,\ell}\bdot\Bar{\mathbf{J}}_{\omega,\ell}
    -\eta_{h,\ell}\grad\Bar{\mathbf{u}}_{\phi,\ell}\bdots\Bar{\mathbf{u}}_{\omega,\ell}
    \Bar{\mathbf{u}}_{\omega,\ell}\cr 
    && \hspace{35pt} +\eta_{h,\ell}\grad\Bar{\mathbf{u}}_{\omega,\ell}
    \bdots(\tau_\ell(\bv_\omega,\bv_\omega) +\tau_\ell(\bv_\phi,\bv_\omega)+\tau_\ell(\bv_\omega,\bv_\phi))
    \label{cgwEulerRotationalEq2}
\end{eqnarray}
where $\Bar{\mathbf{J}}_{\omega,\ell}$ is given by the formula \eqref{Jom-def} with $\nu\mapsto 0,$ 
$(\bv_\omega^\nu,p_\omega^\nu)\mapsto (\bv_\omega,p_\omega).$ 
\end{proposition}
\noindent Our next main result then relates the spatial cascade of rotational energy toward the wall
and the skin friction term in the inviscid Josephson-Anderson relation \eqref{dJA-zero}. However,
the statement is more complex than our previous Theorem \ref{theorem2} for interaction energy,
because it involves also the quantity $\eta_{h,\ell}\grad\Bar{\mathbf{u}}_{\omega,\ell}\bdots\tau_\ell(\bv_\omega,\bv_\omega)$
which represents scale-cascade of rotational energy and as well the anomalous energy dissipation $Q$
in the rotational wake. Thus, we find: 
\begin{theorem}\label{theorem5} 
    Let $(\bv_\omega,p_\omega)$ be the limiting weak solutions of the inviscid equation \eqref{E-omega-mom2}
obtained in Theorem \ref{theorem4}. 
For $0<\ell<h$, and $\forall\varphi\in\Bar{D}(\Bar{\Omega}\times(0,T))$ with $\psi = \varphi|_{\partial B}$,
    \begin{eqnarray} 
        &-&\lim_{h,\ell\to0}\int_0^T\int_{\Omega}\varphi\left(\grad\eta_{h,\ell}\bdot\Bar{\mathbf{J}}_{\omega,\ell} + \eta_{h,\ell}\grad\Bar{\mathbf{u}}_{\omega,\ell}\bdots\tau_\ell(\bv_\omega,\bv_\omega)\right)\,dV\,dt \cr 
        && \hspace{80pt} = -\langle\bv_\phi\bdot \btau_w,\psi\rangle
        + \langle Q,\varphi \rangle. 
    \lb{Jom-bal}     
    \end{eqnarray} 
    % and we obtain again 
\end{theorem}

The result \eqref{Jom-bal} is not entirely satisfactory, because one would expect that the first 
and second terms on the lefthand side equal respectively the first and second terms 
on the righthand side, and not only when added together. We can obtain this physically more natural 
statement, but only by assuming conditions under which the first terms of each side of \eqref{Jom-bal}
vanish: 
\begin{theorem}\label{theorem6} 
Let $(\mathbf{u}_\omega, p_\omega)$ be the limiting weak solutions of the inviscid equation \eqref{E-omega-mom2}
obtained in Theorem \ref{theorem4}, with $\mathbf{u}_\omega\in L^3((0,T),L^3_{\text{loc}}(\Omega))\cap L^3((0,T),L^3(\Omega_{\epsilon}))$ and $p_\omega\in L^{\frac{3}{2}}((0,T),L_{\text{loc}}^{\frac{3}{2}}(\Omega))\cap L^{\frac{3}{2}}(0,T;L^{\frac{3}{2}}(\Omega_{\epsilon}))$ for some $\epsilon>0$. 
Assume further the stronger near-wall boundedness property 
            \begin{align}
                \mathbf{u}_\omega \in L^3((0,T),L^{\infty}(\Omega_{\epsilon})),\;\;\;\;
                p_{\omega}\in L^{\frac{3}{2}}((0,T), L^{\infty}(\Omega_{\epsilon}))
                \label{wallboundednessL3}
            \end{align}
and the no-flow-through condition at the boundary in the sense
      \begin{align}
            \lim_{\delta\to0}\norm{\mathbf{n}\bdot\mathbf{u}_\omega}_{L^3((0,T),L^{\infty}(\Omega_{\delta}))} = 0. 
            \label{wallnormalL3}
        \end{align}
Then $\bv_\phi\bdot\btau_w=0$ in $D'((\partial B)_T)$ and for all $\textbf{Ext}\in {\mathcal E},$
\begin{align}
    \lim_{h,\ell\to0}\mathbf{Ext}^*\left[\grad\eta_{h,\ell}\bdot\Bar{\mathbf{J}}_{\omega,\ell}\right]=0.
\end{align}
In consequence, an inviscid form of the balance equation \eqref{Eom-loc} for rotational energy 
holds distributionally, in the sense that for all $\varphi\in \Bar{D}(\Bar{\Omega}\times (0,T))$, 
   \begin{eqnarray} 
        && %-\int_{\Omega}\frac{1}{2}\varphi|\mathbf{u}_{\omega}|^2(0,x)d\mathbf{x}
        -\int_0^T\int_{\Omega}\frac{1}{2}\partial_t\varphi\,|\mathbf{u}_{\omega}|^2\,dV\, dt
        -\int_0^T\int_{\Omega}\grad\varphi\bdot\left[\frac{1}{2}|\mathbf{u}_{\omega}|^2\mathbf{u}
        +p_{\omega}\mathbf{u}_{\omega}\right]
        dV\,dt\label{fgRotaionalLimit1}\cr
        && \hspace{50pt} =-\langle Q,\varphi \rangle  -\int_0^T\int_{\Omega}\varphi\grad\mathbf{u}_{\phi}\bdots\mathbf{u}_{\omega}\mathbf{u}_{\omega}\,dV\,dt
        \label{fgRotationalLimit3}
    \end{eqnarray}
and furthermore     
\begin{align}
    -\lim_{h,\ell\to0}\int_0^T\int_{\Omega}\varphi
    \, \eta_{h,\ell}\grad\Bar{\mathbf{u}}_{\omega,\ell}
    \bdots \tau_\ell(\bv_\omega,\bv_\omega)\, dV\,dt = \langle Q,\varphi\rangle
    \label{roughnessDissipative}
\end{align}
\end{theorem}

\begin{remark} 
The results \eqref{fgRotationalLimit3} and \eqref{roughnessDissipative} of Theorem \ref{theorem6} 
correspond exactly to those in Proposition 2 of Duchon \& Robert \cite{duchon2000inertial}, 
who considered $L^3$ weak solutions of incompressible Euler equation on a torus and derived 
a local energy balance containing a term due to  ``inertial energy dissipation,'' analogous 
to the lefthand side of our Eq.\eqref{roughnessDissipative}. Furthermore, the equality in 
Eq.\eqref{roughnessDissipative} corresponds to the result in Proposition 4 of \cite{duchon2000inertial},
stating that the ``inertial dissipation'' coincides with the zero-viscosity limit of the 
viscous energy dissipation, assuming convergence to the weak Euler solution in strong $L^3$ sense 
as the viscosity tends to zero. 
\end{remark}

\vspace{-20pt} 
\textcolor{black} {
\begin{remark}
Theorem  \ref{theorem6} implies a strong-weak uniqueness result for any viscosity weak 
solution $\bu_\omega\in L^\infty((0,T),L^2(\Omega)).$ In fact, the quantity 
\be E_\omega(t):=\frac{1}{2}\|\bu_\omega(t)\|^2_{L^2(\Omega)}
=\frac{1}{2}\|\bu(t)-\bu_\phi(t)\|^2_{L^2(\Omega)}\ee 
coincides with the ``relative energy'' used by Wiedemann \cite{wiedemann2018weak} 
to prove  strong-weak uniqueness. We do not give full details here, but just sketch main ideas. 
Under the assumptions of Theorem \ref{theorem4} only we obtain that for a.e. $t\in (0,T)$ 
\be E_\omega(t)-E_\omega(0)\leq \langle \bu_\phi\bdot\btau_w,1\rangle
-\int_0^t \int_\Omega \grad\bu_\phi\bdots\bu_\omega\otimes\bu_\omega\,dV\,ds \lb{Ebd} \ee 
This inequality follows from the result \eqref{fgRotationalLimit2} of Theorem \ref{theorem4}
by substituting $\varphi=\chi^\epsilon_{[0,t]}\chi^\eta_{U_R}$ where $\chi^\epsilon_{[0,t]}$
is an $\epsilon-$smoothed characteristic function of $[0,T]$ and $\chi^\eta_{U_R}$ is an 
$\eta-$smoothed characteristic function of $U_R=\Omega\cap B(\bzed,R)$ and by then taking 
the limits $\epsilon\to 0$ and $R\to\infty.$ Therefore, if we use the consequence 
$\bu_\phi\bdot\btau_w=0$ of Theorem \ref{theorem6}, then it easily follows from \eqref{Ebd} that
\be E_\omega(t)\leq E_\omega(0) + C\rho\int_0^t \|\grad \bu_\phi(s)\|_{L^\infty(\Omega)}
E_\omega(s)\, ds \ee 
and thus by Gronwall inequality
\be E_\omega(t)\leq E_\omega(0)\exp\left(C\rho\int_0^t  \|\grad\bu_\phi(s)\|_{L^\infty(\Omega)}
\, ds\right). \ee
In particular, we see that if $\bu_0=\bu_\phi,$ then $\bu(t)=\bu_\phi$ for a.e $t\in [0,T].$ 
Note that \eqref{Ebd} can be derived also when initial data $\bu_0^\nu\to \bu_\phi$ strongly 
in $L^2(\Omega),$ by taking $\varphi=\chi^\epsilon_{[0,t]}\chi^\eta_{U_R}$ in \eqref{rotationTest}
below for $\nu>0,$ and then taking the successive limits $\epsilon\to 0,$ $\nu\to 0,$ and $R\to\infty.$
Thus, $\bu_\phi\bdot\btau_w=0$ implies strong-weak uniqueness even if the initial data satisfy 
$\bu_0^\nu=\bzed$ at $\partial \Omega$ because of a thin boundary layer. 
\end{remark}
} 

\vspace{-15pt} 
\section{Preliminaries}\label{prelim} 

We present briefly here some background material needed for our proofs. 

\iffalse
\subsection{Manifolds and Distributions}.

See section 2 in \cite{quan2022inertial} for a detailed discussion. 

Since $\partial \Omega=\partial B$ is a compact $C^{\infty}$ submanifold of $\Omega$, there exists $\eta(\Omega)>0$ 
such that $\Omega_{\epsilon}$ for any $\epsilon<\eta(\Omega)$ is a neighborhood of $\partial B\subset \Omega$,
\begin{enumerate}[(i)]
    \item for any $\bx\in\bar{\Omega}_{\epsilon}$, the function $x\mapsto d(\bx)$ is $C^{\infty}(\bar{\Omega}_{\epsilon})$
    \item for any $\bx\in\bar{\Omega}_{\epsilon}$, there exists a unique point $\pi(\bx)\in\partial B$ such that
        \begin{align}
            d(\bx) = |\bx - \pi(\bx)|, \;\;\;\; \grad d(\bx) = \bn(\pi(\bx))
        \end{align}
\end{enumerate}
The unit normal vector $\bn$, projection map $\pi$ and distance function $d$ are all spatially smooth. 
\fi

\subsection{Extension operators.} Extension operators of the type discussed in the introduction are defined more 
precisely by 
\iffalse
In order to study the limits of $\bu_{\phi}\bdot(\nu\grad\btimes\bomega^{\nu})$ and the spatial flux of energy $\bar{\bJ}_{\phi.\ell},\bar{\bJ}_{\omega,\ell}$ as distributions on $(\partial B)_T$, we simplify the approach in \cite{quan2022inertial} and introduce the concept of extension operator:
\fi
\begin{align}
    \ext:\psi\in D((\partial B)_T) \mapsto \varphi\in\Bar{D}(\Bar{\Omega}\times(0,T))
\end{align}
as linear and continuous operators, satisfying pointwise equality $\phi|_{(\partial B)_T}=\psi.$ 
\iffalse
The codomain $D(\mathbb{R}^3\times(0,T))$ is equipped with the usual seminorms, which define the Fréchet topology. 
\fi
Here $\textbf{Ext}$ is continuous in the sense that
for a multi-index $\alpha = (\alpha_1, \alpha_2, \alpha_3,\alpha_4)$ with $|\alpha| \le N$, $\forall (\bx,t)\in\Omega\times(0,T)$
\begin{align}
    |D^{\alpha}\varphi(\mathbf{x},t)|
    =|D^{\alpha}\textbf{Ext}(\psi)(\mathbf{x},t)|\lesssim \sup_{i\in I}p_{N,m,i}(\psi)\label{extensionBound}
\end{align}
where $\lesssim$ denotes inequality with constant prefactor depending on the domain $(\partial B)_T$ 
and the extension operator $\textbf{Ext}$. Here, $(p_{N,m,i})$ are the seminorms associated with the Fréchet space $D((\partial B)_T)$, where $i\in I$ index a system of smooth charts $\cup_{i\in I}(V_{i},\phi_{i}),$ 
$\phi_i: V_i\subset (\partial B)_T\to {\mathbb R}^3,$ and $m$ is the index of a fundamental increasing sequence 
$(K_m^{(i)})$ of compact subsets of $\phi_i(V_i).$ 
See section 2 in \cite{quan2022inertial}. 

We denote by $\mathcal{E}$ the set of all such linear and continuous operators from $D((\partial B)_T)$ to 
\iffalse
$D(\mathbb{R}^3\times(0,T))$, with range 
\fi
$\bar{D}(\bar{\Omega}\times(0,T))$.
We show that $\cE$ is non-empty by constructing an extension operator explicitly.
For $0 < \epsilon < \eta(\Omega)$, we define 
% construct one extension operator in the same way as in section 2 of \cite{quan2022inertial}, but with scalar test functions on $(\partial B)_T$ instead of smooth sections.:
\begin{align}
    \ext^0: \psi\in D((\partial B)_T) 
    \mapsto \varphi \in \Bar{D}(\Bar{\Omega}\times (0,T)) 
\end{align}
by means of the explicit expression 
\begin{align}
    \varphi(\mathbf{x},t) = 
        \begin{cases} \exp\left(-\frac{d(\mathbf{x})}{\epsilon-d(\mathbf{x})}\right)
            \psi(\pi(\mathbf{x}),t), 
            & d(\mathbf{x})<\epsilon\\
            0 & d(\mathbf{x})\ge \epsilon   
        \end{cases}\label{ext0}
\end{align}
for any $0<\epsilon<\eta(\Omega)$.
$\ext^0$ is clearly linear in $\psi$ by definition. 
$\varphi$ is smooth by the smoothness of distance function $d$ and projection $\pi$.
One can easily obtain the bound (\ref{extensionBound}) for $\ext^0$ by product rule and chain rule in calculus.

We use extension operators to identify $F\in D'((0,T), C^{\infty}(\bar{\Omega}))$ with distributions on the 
space-time boundary $(\partial B)_T$. 
For each $\ext\in \cE$, we define
\begin{equation}
    \begin{aligned}
        \ext^*: D'((0,T), C^{\infty}(\bar{\Omega})) &\to D'((\partial B)_T)\\
        F &\mapsto \ext^*(F)
    \end{aligned}
\end{equation}
as follows:
\begin{align}
    \langle\ext^*(F),\psi\rangle := \langle F,\ext(\psi)\rangle
\end{align}
for all $\psi\in D((\partial B)_T)$.
The linearity and continuity properties of $\ext^*(F)$ follow from those of $\ext\in\cE$, so that $\ext^*(F)\in D'((\partial B)_T)$. 
This identification of functions on the whole domain with distributions in 
$D'((\partial B)_T)$ depends of course on the choice of the extension operator $\ext$.

\subsection{Potential flow solution}\lb{sec:pot} 
The potential-flow solution $\bu_{\phi}$ of the incompressible Euler equation
\iffalse
\begin{align}
    \partial_t\bu_{\phi} + \grad\bdot(\bu_{\phi}\bu_{\phi}+p_{\phi}\bI) = 0, \;\;\;\; \grad\bdot\bu_{\phi} = 0
\end{align}
\fi
is given by $\bu_{\phi} = \grad\phi$ for the velocity potential $\phi$ which solves the following Neumann 
problem for the Laplace equation:
\begin{equation}
    \begin{aligned}
        &\Delta\phi = 0\\
        &\frac{\partial\phi}{\partial n}\bigg|_{\partial B} = 0, \;\;\;\; \phi\underset{|\mathbf{x}|\to\infty}{\sim} V(t)x
    \end{aligned}\label{LaplaceEq}
\end{equation}
allowing for the moment a time-dependent external velocity field 
$\bV(t)=V(t)\hat{\bx}.$ The Euler pressure $p_\phi$ is then obtained from the Bernoulli equation
\be \partial_t\phi+ \frac{1}{2}|\grad\phi|^2+p_\phi=c(t) \ee
and is unique up to the arbitrary space-independent constant $c(t).$ For our proofs below it is essential
that $\phi\in C^\infty(\bar{\Omega}\times (0,T))$ and thus possesses smoothness up to the boundary $\partial \Omega.$ 
In the general case of smoothly time-dependent velocity $\bV(t),$  this property will depend upon differentiability 
of the solution of the Neumann problem \eqref{LaplaceEq} in the boundary data. Here, for simplicity, we shall consider
only time-independent flow at infinity, so that the required smoothness follows from the following result:  

\begin{proposition}\lb{prop:pot} 
The Neumann problem (\ref{LaplaceEq}) of the Laplace equation for the case $\bV(t)\equiv\bV$
has a unique time-independent solution $\phi\in C^{\infty}(\Bar{\Omega}).$
\end{proposition}

\begin{proof}
This result is essentially classical but we provide a proof here for completeness. 
By the mapping $\Tilde{\phi}:= \phi - Vx,$ any solution $\phi$ of \eqref{LaplaceEq}
is transformed into a solution $\tilde{\phi}$ of the following inhomogeneous Neumann problem: 
\begin{equation}
    \begin{aligned}
        &\Delta\Tilde{\phi} = 0\\
        &\frac{\partial\Tilde{\phi}}{\partial n}\bigg|_{\partial B} = -V\hat{\bx}\bdot\hat{\bn}, \;\;\;\; \Tilde{\phi}\underset{|\mathbf{x}|\to\infty}{\sim} 0
    \end{aligned}\label{LaplaceEq1}
\end{equation}
From classical potential theory, a unique solution $\Tilde{\phi}\in C^2(\Omega)\cap C^1(\bar{\Omega})$ of (\ref{LaplaceEq1}) exists. See e.g. \cite{medkova2018laplace}, Theorem 6.10.6; \cite{neudert2001asymptotic}, Theorem 2.1. 
Since $\Tilde{\phi}$ is harmonic in the open set $\Omega$, necessarily $\Tilde{\phi}\in C^{\infty}(\Omega)$. See e.g. \cite{evans2010partial}. Thus, $\phi = \Tilde{\phi} + Vx$ is also a smooth function in $C^{\infty}(\Omega)$
and in $C^1(\bar{\Omega}).$

It then suffices to show $\phi$ is smooth in a neighborhood of the boundary $\partial\Omega$.
Since $B$ is compact in $\mathbb{R}^3$, $B\subsetneq B(\bzed,R)$ for some large $R>0$.
Define $U := B(\bzed,R)\cap\Omega$, which has $C^{\infty}$ smooth boundary $\partial U=\partial \Omega\cup \partial B(\bzed,R)$. 
For some sufficiently small $r,$ we have $\Omega_{2r}\subset U$.
To further localize near the boundary, we consider a smooth non-increasing function $\bar{\theta}:\mathbb{R}\to [0,1]$ such that $\bar{\theta} = 1$ in $(-\infty,r]$ and $0$ in $[2r,+\infty)$. We denote $\xi(\bx):=\bar{\theta}(d(\bx))$. As we  reviewed in the introduction, the distance function $d$ is $C^{\infty}$ because of $C^{
\infty}$ smoothness of $\partial\Omega$. Thus, $\xi\in C^\infty(\bar{\Omega}),$ is supported in $\Omega_{2r}$ and $\xi\equiv 1$ in $\Omega_r$.
Since $\phi|_{U}\in C^1(\bar{U})$, then $\xi\phi\in H^1(U)$ is a weak solution of the following homogeneous Neumann problem
for the Poisson equation: 
\begin{align}
    \Delta(\xi\phi) &= \phi\Delta\xi + 2\grad\phi\bdot\grad\xi := f\;\;\;\;\text{ in } U\\
    \frac{\partial(\xi\phi)}{\partial n} &= 0 \;\;\;\;\text{ on } \partial U. % \phi\grad\xi\cdot\bn := g
\end{align}
By definition, $f\in C^{\infty}(\bar{U})$, since $\phi$ is $C^{\infty}$ in $\Omega$ and $\Delta\xi = \grad\xi = 0$ in $\bar{\Omega_r}.$ Furthermore $\int_{U} f\, dV=0.$
% \iffalse
% and $g\in H^{\frac{1}{2}}(\partial U)$ with $\norm{g}_{H^{\frac{1}{2}}(\partial U)} \lesssim \norm{\phi}_{H^1(U)}$, which is obtained from the continuity of trace operator. 
% \fi 
Since $\partial U$ is $C^{\infty}$, 
we can deduce from boundary elliptic regularity theory, such as Theorem 4 in section 4.2 of \cite{mikhailov1978pdes}, that 
\begin{align}
    \norm{\phi}_{H^{k+2}(\Omega_{r})}\le\norm{\xi\phi}_{H^{k+2}(U)} &\lesssim \norm{f}_{H^k(U)} <\infty
\end{align}
for all integers $k\ge1$.
Here, the inequalities $\lesssim$ above with constant prefactors depend only on $\Omega, R, r$.
% Since $\phi\in C^{\infty}(\Omega)$, we deduce that $\phi\in H^2(U)$. Applying this argument inductively for any integer $k\ge 0$, with multi-index $|\alpha| = k$, on the Neumann problem of $\xi D^{\alpha}\phi$, yields that $\phi\in H^{k+2}(U)$. The $C^{\infty}$ smoothness of $\partial\Omega$ also comes in when we analyze the normal derivative of $D^{\alpha}\phi$ on the boundary.
The Sobolev embedding theorem implies $\phi\in C^{k}(\bar{\Omega}_r)$ for any integer $k\ge0$, and thus $\phi\in C^{\infty}(\bar{\Omega}_r)$. Together with the interior smoothness of $\phi$, we conclude that $\phi\in C^{\infty}(\bar{\Omega})$. \hspace{170pt}  \qed
\end{proof}

\section{Proof of Theorems \ref{theorem1}-\ref{theorem3}}\label{thm13}
\subsection{Proof of Theorem \ref{theorem1}}
We take an arbitrary $\varphi\in\Bar{D}(\Bar{\Omega}\times(0,T))$ and let $\psi = \varphi|_{\partial B}$. Testing the interaction energy equation (\ref{Eint-loc}) against $\varphi$ yields
\begin{equation}
    \begin{aligned}
        &-\int_0^T\int_{\Omega}\partial_t\varphi(\mathbf{u}_{\omega}^{\nu}\bdot\mathbf{u}_{\phi})\,dV\,dt\\
        &-\int_0^T\int_{\Omega}\grad\varphi\bdot\left[(\mathbf{u}_{\omega}^{\nu}\bdot\mathbf{u}_{\phi})\mathbf{u}^{\nu}
        +\frac{1}{2}|\mathbf{u}_{\phi}|^2\mathbf{u}_{\omega}^{\nu}+p_{\omega}^{\nu}\mathbf{u}_{\phi}
        +p_{\phi}\mathbf{u}_{\omega}^{\nu}\right]\,dV\,dt\\
        &= \int_0^T\int_{\Omega}\varphi(\bu_{\phi}\bdot\nu\triangle \bv^{\nu})\,dV\, dt
        +\int_0^T\int_{\Omega}\varphi(\grad\mathbf{u}_{\phi}\bdots\mathbf{u}_{\omega}^{\nu}\otimes\mathbf{u}_{\omega}^{\nu})
        \,dV\,dt
    \end{aligned}\label{visc_total_interaction}
\end{equation}
As $\varphi\in \Bar{D}(\Bar{\Omega}\times(0,T))$, there exists a compact subset $K_{\varphi}\subset\mathbb{R}^3$ such that 
\begin{align}
    \text{supp}(\varphi)\subset K_{\varphi}\times(0,T)\subset\Bar{\Omega}\times(0,T)
\end{align}
\iffalse
We deduce from $\bu_{\phi} = \grad\phi$ that $\grad\bdot(\bu_{\phi}\btimes\nu\bomega^{\nu}) = -\bu_{\phi}\bdot(\nu\grad\btimes\bomega^{\nu})$.
Using this identity and integration by parts yields
\begin{align}
    -\int_0^T&\int_{\Omega}\varphi(\bu_{\phi}\bdot(\nu\grad\btimes\bomega^{\nu}))\,dV\, dt\\ \nonumber
    &=-\int_0^T\int_{\partial B}\psi\mathbf{u}_{\phi}\bdot(\nu\bomega^{\nu}\btimes\mathbf{n})\,dS\,dt
        - \int_0^T\int_{\Omega}\grad\varphi\bdot(\nu\mathbf{u}_{\phi}\btimes\bomega)\,dV\,dt
\end{align}
\fi
It follows by Green's identity and stick b.c. for the velocity $\bv^\nu$ that
\begin{eqnarray} 
    && -\int_0^T\int_{\Omega}\varphi\bu_{\phi}\bdot \nu\triangle\bv^{\nu}\,dV\, dt\cr 
    && \hspace{30pt} =-\int_0^T\int_{\partial B}\psi\mathbf{u}_{\phi}\bdot\btau_w^\nu\,dS\,dt
        - \int_0^T\int_{\Omega} \nu \triangle(\varphi\mathbf{u}_{\phi})\bdot\bv^\nu\,dV\,dt
    \label{u-Greens}     
\end{eqnarray} 
with $\btau_w^\nu=\nu\partial\bv^\nu/\partial n.$
\iffalse
We obtain an upper bound by Hölder's inequality
\begin{align}
    &\int_0^T\int_{\Omega}|\grad\varphi\bdot(\nu\mathbf{u}_{\phi}\btimes\bomega^{\nu})|d\mathbf{x}dt 
    = \int_0^T\int_{\Omega}|\sqrt{\nu}\mathbf{u}_{\phi}\bdot(\sqrt{\nu}\bomega^{\nu}\btimes\grad\varphi)|d\mathbf{x}dt\\ \nonumber
    &\hspace{30pt} \le\left(\int_{\text{supp}(\varphi)}\nu|\mathbf{u}_{\phi}|^2d\mathbf{x}dt\right)^{\frac{1}{2}} \cdot \left(\int_0^T\int_{\Omega}|\grad\varphi|^2\nu|\bomega^{\nu}|^2d\mathbf{x}dt\right)^{\frac{1}{2}}\label{interactionBoundary}
\end{align}
It follows from the anomalous dissipation condition (\ref{viscDissLimit}) that this upper bound converges to 0 as $\nu\to0$.
\fi
We obtain an upper bound by Cauchy-Schwartz
\begin{align}
    &\left|\int_0^T\int_{\Omega} \nu \triangle(\varphi\mathbf{u}_{\phi})\bdot\bv^\nu\,dV\,dt\right|
    \\ \nonumber
    &\hspace{30pt} \le \nu \left(\int_0^T \int_{K_\varphi} |\mathbf{u}^\nu|^2
    \, dV\,dt\right)^{\frac{1}{2}} \cdot \left(\int_0^T\int_{\Omega}|\triangle(\varphi\bv_\phi)|^2\, dV\,dt\right)^{\frac{1}{2}}\label{interactionBoundary}
\end{align}
which vanishes as $\nu\to 0.$ By Proposition \ref{prop:pot}, 
$\bu_{\phi}$ is smooth everywhere in $\Bar{\Omega}\times(0,T)$ and tangential at the boundary $\partial B$, so that 
$\varphi\mathbf{u}_{\phi}\in \Bar{D}(\Bar{\Omega}\times(0,T),\mathbb{R}^3)$, with $(\varphi\mathbf{u}_{\phi})|_{\partial B} = \psi\mathbf{u}_{\phi}|_{\partial B}$ and $\mathbf{u}_{\phi}|_{\partial B}\bdot\mathbf{n}\equiv 0$.
We can thus identify $\psi\mathbf{u}_{\phi}|_{\partial B}$ as a smooth section of the cotangent bundle $D((\partial B)_T,\mathcal{T}^*(\partial B)_T)$, via its graph as the natural embedding
\begin{align}
    \psi\mathbf{u}_{\phi}|_{\partial B}(\bx,t) \mapsto (\bx, t, \psi\mathbf{u}_{\phi}|_{\partial B}(\bx, t), 0). 
\end{align}
Thus by the convergence of skin friction (\ref{tau-lim}), we have
\begin{eqnarray} 
    \int_0^T\int_{\partial B}\psi\mathbf{u}_{\phi}\bdot \btau_w^\nu\,dS\,dt = \langle\btau_w^{\nu},\psi\mathbf{u}_{\phi}|_{\partial B}\rangle
    &&\xrightarrow{\nu\to0}\langle\btau_w,\psi\mathbf{u}_{\phi}|_{\partial B}\rangle\cr
    && \hspace{20pt} = \langle\bu_{\phi}\bdot\btau_w, \psi\rangle.  \label{uphi-tau1} 
\end{eqnarray} 
Combining \eqref{u-Greens}-\eqref{uphi-tau1} gives 
\be -\int_0^T \int_\Omega \varphi\bu_\phi\bdot\nu\triangle\bu^\nu\,dV\,dt
\xrightarrow{\nu\to 0} \langle\bu_{\phi}\bdot\btau_w, \psi\rangle. \ee
Note that for every extension operator $\ext\in\cE$, we can define $\bu_{\phi}\bdot(\nu\grad\btimes\bomega^{\nu})
=-\bu_{\phi}\bdot(\nu\triangle\bv^{\nu})$ as a distribution in $D'((\partial B)_T)$ via 
\begin{align}
    \langle\ext^*(\bu_{\phi}\bdot(\nu\grad\btimes\bomega^{\nu})),\psi\rangle 
    &=-\int_0^T\int_{\Omega}\varphi(\bu_{\phi}\bdot(\nu\triangle \bu^{\nu}))\, dV\, dt
\end{align}
for all $\psi\in D((\partial B)_T)$ with $\varphi = \ext(\psi)\in \bar{D}(\bar{\Omega}\times(0,T))$. Thus, as $\nu\to0$
\begin{align}
    \ext^*(\bu_{\phi}\bdot(\nu\grad\btimes\bomega^{\nu})) \to \bu_{\phi}\bdot\btau_w\;\; in \;\; D'((\partial B)_T)
\end{align}

Just as in Theorem 1 and Lemma 1 of \cite{quan2022inertial}, we can deduce from the assumptions (\ref{L2Conv}-\ref{pBBound}) that 
\begin{align}
    \bu_{\omega}\in L^2((0,T), L^2(\Omega_{\epsilon})), \;\;\;\; p_{\omega}\in L^1((0,T), L^1(\Omega_{\epsilon}))
\label{0nwL2}     
\end{align}
and as $\nu\to0$
\begin{align}
    \norm{(\bu_{\omega}^{\nu}- \bu_{\omega})\otimes \bu_{\phi}}_{L^1((0,T) \times K_{\varphi}))}
        &\le
        \norm{\bu_{\omega}^{\nu} - \bu_{\omega}}_{L^2((0,T),L^2(K_{\varphi}))}\norm{\bu_{\phi}}_{L^2((0,T),L^2(K_{\varphi}))}\nonumber\\
        &\to 0
\end{align}
This implies that as $\nu\to0$
\begin{align}
    -\int_0^T\int_{\Omega}\partial_t\varphi(\mathbf{u}_{\omega}^{\nu}\bdot\mathbf{u}_{\phi})
    \,dV\,dt
    \to
    -\int_0^T\int_{\Omega}\partial_t\varphi(\mathbf{u}_{\omega}\bdot\mathbf{u}_{\phi})
    \,dV\,dt
\end{align}
All of the terms in (\ref{visc_total_interaction}) converge along the same sequence as $\nu\to0$ 
to their inviscid limit, by very similar arguments as above. Thus, we obtain the local balance equation (\ref{weakInteractionLimit1}) for interaction energy. 

Finally, we can write the Josephson-Anderson relation \eqref{dJA-bd} in the form 
\begin{eqnarray} 
&& 
\int_0^T dt \chi(t) {\bf F}^\nu(t)\bdot\bV\,dt = \cr
&& \hspace{30pt} -\rho \int_0^T
\int_\Omega  \chi \grad\bv_\phi\bdots\bv_\omega^\nu\otimes\bv_\omega^\nu\, dV\,dt + 
\rho\langle\bv_\phi\bdot\btau_w^\nu,\chi\otimes 1\rangle
\lb{dJA-bd-2} \end{eqnarray} 
after smearing with $\chi\in {\mathcal D}((0,T)).$ Because the velocity potential $\phi$ 
in Proposition \ref{prop:pot} is harmonic in ${\mathbb R}^3\backslash B,$ it follows that 
$|\grad \bv_\phi|=O(|\bx|^{-3})$ as $|\bx|\to\infty;$ e.g. see \cite{medkova2018laplace},
Proposition 2.17.3. Since also $\phi\in C^\infty(\bar{\Omega}),$ it holds immediately that 
$\grad\bv_\phi\in L^\infty((0,T)\times\Omega).$ Together with assumption 
\eqref{L2Conv-str}, we get that as $\nu\to 0$
$$ \int_0^T
\int_\Omega  \chi \grad\bv_\phi\bdots\bv_\omega^\nu\otimes\bv_\omega^\nu\, dV\,dt
\to \int_0^T
\int_\Omega  \chi \grad\bv_\phi\bdots\bv_\omega\otimes\bv_\omega\, dV\,dt
$$
and the righthand side is well-defined because $\bv_\omega\in L^2((0,T),L^2(\Omega)).$
Using also the convergence (\ref{tau-lim}) of the skin friction $\btau_w^\nu,$ we 
conclude that 
\begin{eqnarray*} 
&& \lim_{\nu\to 0}
\int_0^T dt \chi(t) {\bf F}^\nu(t)\bdot\bV\,dt = \cr
&& \hspace{60pt} -\rho \int_0^T
\int_\Omega  \chi \grad\bv_\phi\bdots\bv_\omega\otimes\bv_\omega\, dV\,dt + 
\rho\langle\bv_\phi\bdot\btau_w,\chi\otimes 1\rangle. 
\qquad \qquad \quad \qed
\lb{dJA-bd-3} \end{eqnarray*}

\subsection{Proof of Theorem \ref{theorem2}}
% TODO: change all the 'l' to '\ell'
% TODO: state the lemma
% TODO: convergence of RHS
% TODO: compare with tau-lim
% TODO: add subscript \ell; \tau_{\ell}
\iffalse
Since the inviscid limit $\bu_{\phi}\bdot\btau_w$ is defined as a distribution on $(\partial B)_T$, we need to first identify the spatial flux of interaction energy with similar distributions. As discussed in Section \ref{prelim}. We can define 
$\ext^*(\grad\eta_{h,\ell}\bdot\Bar{\mathbf{J}}_{\phi,\ell})\in D'((\partial B)_T)$ with some $\ext\in\cE$ as
\begin{align}
    \langle\ext^*(\grad\eta_{h,\ell}\bdot\Bar{\mathbf{J}}_{\phi,\ell}), \psi\rangle
    &= \langle\grad\eta_{h,\ell}\bdot\Bar{\mathbf{J}}_{\phi,\ell}, \ext(\psi)\rangle\\
    &= \int_0^T\int_{\Omega}\ext(\psi)(\grad\eta_{h,\ell}\bdot\Bar{\mathbf{J}}_{\phi,\ell})d\mathbf{x}dt
\end{align}
for all $\psi\in D((\partial B)_T)$. Fix an arbitrary $\psi\in D((\partial B)_T)$ and let $\varphi = \ext(\psi)$ for some arbitrary $\ext\in \cE$. Since $\varphi\in \Bar{D}(\Bar{\Omega}\times(0,T))$, there exists some compact set $K_{\varphi}$ such that $\text{supp}(\varphi)\subset K_{\varphi}\times(0,T)\subset\Bar{\Omega}\times(0,T)$.
\fi 
Fix an arbitrary $\psi\in D((\partial B)_T)$ and let $\varphi = \ext(\psi)$ for some arbitrary $\ext\in \cE$.
Testing the inviscid balance equation \eqref{cgWindowEulerInteraction} of coarse-grained interaction energy 
against $\varphi$ and rearranging yields
\begin{align}
    -&\langle\ext^*(\grad\eta_{h,\ell}\bdot\Bar{\mathbf{J}}_{\phi,\ell}), \psi\rangle\,=\,\int_0^T\int_{\Omega}\varphi(\grad\eta_{h,\ell}\bdot\Bar{\mathbf{J}}_{\phi,\ell})\,dV\,dt
    \label{cgEulerInteractionLocalEq}\\
    \nonumber
    =&\int_0^T\int_{\Omega}[(\partial_t\varphi)\eta_{h,\ell}(\Bar{\mathbf{u}}_{\omega,\ell}\bdot\Bar{\mathbf{u}}_{\phi,\ell})
    +\grad\varphi\bdot(\eta_{h,\ell}\Bar{\mathbf{J}}_{\phi,\ell})]\,dV\,dt\\\nonumber
    +&\int_0^T\int_{\Omega}\varphi\eta_{h,\ell}\grad\Bar{\mathbf{u}}_{\phi,\ell}\bdots(\tau_{\ell}(\bu_{\omega},\bu_{\omega}) +\tau_{\ell}(\bu_{\phi},\bu_{\omega})+\tau_{\ell}(\bu_{\omega},\bu_{\phi}))\,dV\,dt\\\nonumber
    -&\int_0^T\int_{\Omega}\varphi\eta_{h,\ell}\Bar{\mathbf{u}}_{\omega,\ell}\bdot(\grad\bdot\tau_{\ell}(\bu_{\phi},\bu_{\phi}))\,dV\,dt
    +\int_0^T\int_{\Omega}\varphi\eta_{h,\ell}\Bar{\mathbf{u}}_{\omega,\ell}\Bar{\mathbf{u}}_{\omega,\ell}\bdots\grad\Bar{\mathbf{u}}_{\phi,\ell}\,dV\,dt
\end{align}
Our strategy is to prove that all the terms on the right hand side of (\ref{cgEulerInteractionLocalEq}) converge as $h,\ell\to0$ and then to compare with the balance (\ref{weakInteractionLimit1}) obtained in 
Theorem \ref{theorem1}. 

The proof of the convergence relies on the following lemma: 
\begin{lemma}\label{lemma1}
    If $f_i\in L^{p_i}((0,T),L_{\text{loc}}^{p_i}(\Omega))\cap L^{p_i}((0,T),L^{p_i}(\Omega_{\epsilon}))$ for $i=1,\dots k$, with some integer $k$, and $p_i\in[1,\infty]$ and $\sum_{i=1}^k\frac{1}{p_i}=1$, for some arbitrarily small $\epsilon>0$, then for $0<\ell<h$, and some compact set $K\subset\Bar{\Omega}$
    \begin{align}
        \eta_{h,\ell}\prod_{i=1}^k\Bar{f_i}_{\ell}\xrightarrow[L^{1}((0,T),L^1(K))]{h,\ell\to0}\prod_{i=1}^kf_i
    \end{align}
\end{lemma}
This is a simple extension of Lemma 2 in \cite{quan2022inertial}, easily proved by induction on $k$. \\
Note that $\bu_{\phi}$ is smooth everywhere in $(\partial B)_T$ and (\ref{L2Conv}-\ref{pBBound}) imply that
\begin{align}
    \bu_{\omega}&\in L^2((0,T), L_{\text{loc}}^2(\Omega))\cap L^2((0,T), L^2(\Omega_{\epsilon}))\\
    p_{\omega}&\in L^1((0,T), L_{\text{loc}}^1(\Omega))\cap L^1((0,T), L^2(\Omega_{\epsilon}))
\end{align}
Thus, we obtain by Lemma \ref{lemma1} that
\begin{align}
    \eta_{h,\ell}\Bar{\mathbf{u}}_{\omega,\ell}\bdot\Bar{\mathbf{u}}_{\phi,\ell}
    &\xrightarrow[L^1((0,T),L^1(K_{\varphi}))]{h,\ell\to0}\mathbf{u}_{\omega}\bdot\mathbf{u}_{\phi}\\
    \eta_{h,\ell}\Bar{\mathbf{u}}_{\omega,\ell}\otimes\Bar{\mathbf{u}}_{\omega,\ell}\bdots\grad\Bar{\mathbf{u}}_{\phi,\ell}
    &\xrightarrow[L^1((0,T),L^1(K_{\varphi}))]{h,\ell\to0}\mathbf{u}_{\omega}\otimes\mathbf{u}_{\omega}\bdots\grad\mathbf{u}_{\phi}
\end{align}
and likewise 
\begin{eqnarray}
    &&\int_0^T\int_{K_\varphi} |\eta_{h,\ell}\tau_{\ell}(\bu_{\omega},\bu_{\omega})|\,dV\,dt
    \,\le\,\int_0^T\int_{K_\varphi}\big|\eta_{h,\ell}\overline{(\mathbf{u}_{\omega}\otimes\mathbf{u}_{\omega})}_\ell - \mathbf{u}_{\omega}\otimes\mathbf{u}_{\omega}\big|\,dV\,dt \cr
    &&\hspace{100pt} 
    + \int_0^T\int_{K_\varphi}\big|\mathbf{u}_{\omega}\otimes\mathbf{u}_{\omega}-\eta_{h,\ell}
    \Bar{\mathbf{u}}_{\omega,\ell}\otimes\Bar{\mathbf{u}}_{\omega,\ell}\big|\,dV\,dt\cr 
    &&\hspace{100pt} \xrightarrow{h,\ell\to0}0
\end{eqnarray}
Similarly,
\begin{align}
    \eta_{h,\ell}\tau_{\ell}(\bu_{\omega},\bu_{\phi})\xrightarrow[L^1((0,T),L^1(K_{\varphi}))]{h,\ell\to0} 0,
    \;\;\;\;
    \eta_{h,\ell}\tau_{\ell}(\bu_{\phi},\bu_{\omega})\xrightarrow[L^1((0,T),L^1(K_{\varphi}))]{h,\ell\to0} 0
\end{align}
Thus, all the terms in (\ref{cgEulerInteractionLocalEq}) that contains these cumulants $\tau_{\ell}(\bu_{\omega},\bu_{\omega}), \tau_{\ell}(\bu_{\omega},\bu_{\phi})$, $\tau_{\ell}(\bu_{\phi},\bu_{\omega})$, vanish in the limit as $h, \ell\to0$. We deduce from these vanishing cumulants and Lemma \ref{lemma1} that
\begin{align}
    \eta_{h,\ell}\Bar{\mathbf{J}}_{\phi, \ell}
    \xrightarrow[L^1((0,T),L^1(K_{\varphi}))]{h,\ell\to0}
    \left[(\mathbf{u}_{\omega}\bdot\mathbf{u}_{\phi})\mathbf{u}
    +p_{\omega}\mathbf{u}_{\phi}
    +\left(\frac{1}{2}|\mathbf{u}_{\phi}|^2+p_{\phi}\right)\mathbf{u}_{\omega}\right]
\end{align}
\iffalse
Note that $G$ is a standard mollifier and $\int_{\mathbb{R}^3}(\grad G)_{\ell}(\br)d\br = 0$. The Constantin-E-Titi commutator \cite{constantin1994onsager} identity shows that cumulants can be expressed in terms of increments
\begin{align}
    \btau_{\ell}(\mathbf{f},&\mathbf{g})(\bx)
    = \langle G_{\ell},\delta\mathbf{f}(\cdot;\bx)\otimes\delta\mathbf{g}(\cdot;\bx)\rangle
    - \langle G_{\ell},\delta\mathbf{f}(\cdot;\bx)\rangle
    \otimes \langle G_{\ell},\delta\mathbf{g}(\cdot;\bx)\rangle\label{cet-commutator}
\end{align}
for $\mathbf{f}, \mathbf{g}\in L_{loc}^2(\Omega)$ and for all $x\in\Omega^{\ell}$. Here $\langle\cdot,\cdot\rangle$ denotes the spatial integral over $\mathbb{R}^3$, and $\delta\mathbf{f}(\br, \bx) = \mathbf{f}(\bx+\br) - \mathbf{f}(\bx)$ is the increment.
Using this expression consisting of increments,  we can obtain the following bound from smoothness of $\bu_{\phi}$ 
\begin{align}
    &|\eta_{h,\ell}\grad\cdot\tau_{\ell}(\bu_{\phi},\bu_{\phi})(\bx, t)|\\ \nonumber
        =&\bigg|\eta_{h,\ell}\grad\cdot\bigg\{\langle G_{\ell},\delta\bu_{\phi}(\cdot;\bx)\otimes\delta\bu_{\phi}(\cdot;\bx)\rangle
        - \langle G_{\ell},\delta\bu_{\phi}(\cdot;\bx)\rangle
        \otimes \langle G_{\ell},\delta\bu_{\phi}(\cdot;\bx)\rangle\bigg\}\bigg|\\ \nonumber
        \lesssim &\ell\norm{\grad\bu_{\phi}}_{L^{\infty}((0,T), L^{\infty}(K_{\varphi}))}^2
\end{align}
for every $(\bx, t)\in\Bar{\Omega}\times(0,T)$. 
\fi
To obtain a bound on the term containing $\grad\bdot\tau_{\ell}(\bu_{\phi},\bu_{\phi}),$ we use a
general commutator estimate on gradients of cumulants from \cite{drivas2018onsager}, Proposition 4:  
\begin{eqnarray}
\|\grad\bdot\tau_{\ell}(\bu_{\phi},\bu_{\phi})\|_{L^2(0,T)\times (K_\varphi\backslash\Omega_h))}
&\lesssim & \frac{1}{\ell} \sup_{|\br|<\ell}
\|\delta\bu_{\phi}(\br)\|^2_{L^4(0,T)\times (K_\varphi\backslash\Omega_h))} \cr
&=& O(\ell) 
\end{eqnarray} 
where $\delta f(\br; \bx):= f(\bx+\br) - f(\bx)$ and the last bound follows from smoothness 
of $\bv_\phi.$ Using this estimate and the $L^2$-bound \eqref{0nwL2} on $\bv_\omega,$
we see that the term in the equation (\ref{cgEulerInteractionLocalEq}) containing 
$\grad\bdot\tau_{\ell}(\bu_{\phi},\bu_{\phi})$ also vanishes in the limit as $h,\ell\to0$.

From the $L^1$ convergence of all the integrands on the right hand side of the equation (\ref{cgEulerInteractionLocalEq}), we conclude that 
\begin{align}
    & -\langle\ext^*(\grad\eta_{h,\ell}\bdot\Bar{\mathbf{J}}_{\phi,\ell}),\psi\rangle
    \xrightarrow{h,\ell\to0}\\\nonumber
    &\int_0^T\int_{\Omega}\partial_t\varphi(\mathbf{u}_{\omega}\bdot\mathbf{u}_{\phi})
    \,dV\,dt
    +\int_0^T\int_{\Omega}\varphi\mathbf{u}_{\omega}\mathbf{u}_{\omega}\bdots\grad\mathbf{u}_{\phi}
    \,dV\,dt\\\nonumber
    &-\int_0^T\int_{\Omega} \grad\varphi\bdot
    \left[(\mathbf{u}_{\omega}\bdot\mathbf{u}_{\phi})\mathbf{u}
    +p_{\omega}\mathbf{u}_{\phi}
    +(\frac{1}{2}|\mathbf{u}_{\phi}|^2+p_{\phi})\mathbf{u}_{\omega}\right]
    dV\,dt\\\nonumber
    &\hspace{100pt} =\langle\bv_\phi\bdot\btau_w,\psi\rangle
\end{align}
where the final equality follows by comparison with (\ref{weakInteractionLimit1}).
This convergence for all $\psi\in D'((\partial B)_T)$ yields the stated result.

\subsection{Proof of Theorem \ref{theorem3}}
The proof is similar to that for Theorem 3 in \cite{quan2022inertial}, 
by bounding directly the momentum flux: 
\begin{align}
    \grad\eta_{h,\ell}\bdot\Bar{\mathbf{J}}_{\phi,\ell} &= 
    \theta_{h,\ell}'\mathbf{n}
    \bdot\bigg[(\Bar{\mathbf{u}}_{\omega,\ell}\bdot\Bar{\mathbf{u}}_{\phi,\ell})\Bar{\mathbf{u}}_{\ell}+\Bar{p}_{\omega}\Bar{\mathbf{u}}_{\phi,\ell}
    +\left(\frac{1}{2}|\Bar{\mathbf{u}}_{\phi,\ell}|^2+\Bar{p}_{\phi}\right)\Bar{\mathbf{u}}_{\omega,\ell}\\\nonumber
    &+ (\tau_{\ell}(\bu_{\omega},\bu_{\omega}) +\tau_{\ell}(\bu_{\phi},\bu_{\omega})+\tau_{\ell}(\bu_{\omega},\bu_{\phi}))\bdot\Bar{\mathbf{u}}_{\phi,\ell}\bigg]
\end{align}
which is supported in $\Omega_{h+\ell}\backslash\Omega_h\subset\Omega_{3h}\subset\Omega_{\epsilon}$. 
For all $\mathbf{x}\in\Omega_{h+\ell}\backslash\Omega_h$, a.e. $t\in (0,T)$, we write
\begin{align}
    \mathbf{n}(\pi(\mathbf{x}))\bdot\Bar{\mathbf{u}}_{\ell}(\mathbf{x},t) 
    &= \int_{\mathbb{R}^3}G_{\ell}(\mathbf{r})[\mathbf{n}(\pi(\mathbf{x}))-\mathbf{n}(\pi(\mathbf{x} + \mathbf{r}))]
    \bdot\mathbf{u}(\mathbf{x} + \mathbf{r},t)\, V(d\mathbf{r})\\\nonumber 
    &+\int_{\mathbb{R}^3}G_{\ell}(\mathbf{r})\mathbf{n}(\pi(\mathbf{x} + \mathbf{r}))
    \bdot\mathbf{u}(\mathbf{x} + \mathbf{r},t)\, V(d\mathbf{r}) 
\end{align}
Since $\mathbf{n}\circ\pi$ is smooth in $\bar{\Omega}_{\epsilon}$, $\forall\delta>0$, $\exists\rho=\rho(\delta)>0$ s.t.
\begin{align}
    |\mathbf{n}(\pi(\mathbf{x}))-\mathbf{n}(\pi(\mathbf{x} + \mathbf{r}))|\le\delta
\end{align}
for all $\mathbf{x}\in\Omega_{h+\ell}\backslash\Omega_{h}$ and $|\br|<\ell<\rho$. Then it follows that
\begin{align}
    |\mathbf{n}(\pi(\mathbf{x}))
    \bdot\Bar{\mathbf{u}}_{\ell}(\mathbf{x},t)|
    &\le\delta\norm{\mathbf{u}(t)}_{L^{\infty}(\Omega_{\epsilon})} 
    + \norm{\mathbf{n}\bdot\mathbf{u}(t)}_{L^{\infty}(\Omega_{3h})}
    \label{ubarn-bd}
\end{align}
Obviously 
\begin{align}
    \norm{\mathbf{u}(t)}_{L^{\infty}(\Omega_{\epsilon})} \le \norm{\mathbf{u}_{\omega}(t)}_{L^{\infty}(\Omega_{\epsilon})}+\norm{\mathbf{u}_{\phi}(t)}_{L^{\infty}(\Omega_{\epsilon})}
\end{align}
and thus  $\bu\in L^2((0,T), L^{\infty}(\Omega_{\epsilon})).$ Bounds analogous 
to \eqref{ubarn-bd} hold for $\bv_\omega,$ $\bv_\phi.$ 

Recalling the definition of $\tau_{\ell}(\bu_{\omega}, \bu_{\omega})$ we obtain 
in similar fashion that
\begin{align}
    \tau_{\ell}(\bu_{\omega},\bu_{\omega})\bdot\mathbf{n}(\pi(\mathbf{x}))
    &\le\left(\delta\norm{\mathbf{u}_{\omega}(t)}_{L^{\infty}(\Omega_{\epsilon})} + \norm{\mathbf{n}\bdot\mathbf{u}_{\omega}(t)}_{L^{\infty}(\Omega_{3h})}\right)\norm{\mathbf{u}_{\omega}(t)}_{L^{\infty}(\Omega_{\epsilon})}
\end{align}
and likewise that 
\begin{equation}
    \begin{aligned}
        \tau_{\ell}(\bu_{\phi},\bu_{\omega})\bdot\mathbf{n}(\pi(\mathbf{x}))
        &\le\left(\delta\norm{\mathbf{u}_{\omega}(t)}_{L^{\infty}(\Omega_{\epsilon})} + \norm{\mathbf{n}\bdot\mathbf{u}_{\omega}(t)}_{L^{\infty}(\Omega_{3h})}\right)\norm{\mathbf{u}_{\phi}(t)}_{L^{\infty}(\Omega_{\epsilon})}\\
        \tau_{\ell}(\bu_{\omega},\bu_{\phi})\bdot\mathbf{n}(\pi(\mathbf{x}))
        &\le\left(\delta\norm{\mathbf{u}_{\phi}(t)}_{L^{\infty}(\Omega_{\epsilon})} + \norm{\mathbf{n}\bdot\mathbf{u}_{\phi}(t)}_{L^{\infty}(\Omega_{3h})}\right)\norm{\mathbf{u}_{\omega}(t)}_{L^{\infty}(\Omega_{\epsilon})}
    \end{aligned}
\end{equation}

Consider then some $\psi\in D((\partial B)_T)$ and $\ext\in\cE$, and let $\varphi = \ext(\psi)\in \Bar{D}(\Bar{\Omega}\times(0,T))$.
Using the bounds above, together with $\norm{\theta_{h,\ell}'(d(\mathbf{x}))}_{L^{\infty}}\le\frac{C}{\ell}$ 
and $|\Omega_{h+\ell}\backslash\Omega_h|\le C'\ell$, we obtain that
\begin{equation}
    \begin{aligned}
        \Bigg|\int_0^T\int_{\Omega}\varphi\theta_{h,\ell}'&(\Bar{\mathbf{u}}_{\omega,\ell}\bdot\Bar{\mathbf{u}}_{\phi,\ell})(\Bar{\mathbf{u}}_{\ell}\bdot\bn)\,dV\, dt \Bigg| \\
        &\lesssim\norm{\psi}_{L^{\infty}((\partial B)_T)}\int_0^T\int_{\Omega_{h+l}\backslash\Omega_h}|\theta_{h,\ell}'(\Bar{\mathbf{u}}_{\omega,\ell}\bdot\Bar{\mathbf{u}}_{\phi,\ell})(\Bar{\mathbf{u}}_{\ell}\bdot\bn)|\,dV\,dt\\
        &\lesssim\norm{\psi}_{L^{\infty}((\partial B)_T)}
        \times\norm{\mathbf{u}_{\omega}}_{L^2((0,T),L^{\infty}(\Omega_{\epsilon}))}
        \times\norm{\mathbf{u}_{\phi}}_{L^{\infty}(\Omega_{\epsilon}\times(0,T))}\\
        &\hspace{10pt} \times\left[\delta\norm{\mathbf{u}}_{L^2((0,T),L^{\infty}(\Omega_{\epsilon}))} 
        + \norm{\mathbf{n}\bdot\mathbf{u}}_{L^2((0,T),L^{\infty}(\Omega_{3h}))}\right]
    \end{aligned}
    \label{Jphi-bd1} 
\end{equation}
Similarly,
\begin{equation}
    \begin{aligned}
        \Bigg|\int_0^T\int_{\Omega}&\varphi\theta_{h,\ell}'\left(\frac{1}{2}|\Bar{\mathbf{u}}_{\phi,\ell}|^2+\Bar{p}_{\phi}\right)(\Bar{\mathbf{u}}_{\omega,\ell}\bdot\bn)\,dV\, dt\Bigg|\\
        &\lesssim\norm{\psi}_{L^{\infty}((\partial B)_T)}
        \times\left[\frac{1}{2}\norm{\mathbf{u}_{\phi}}_{L^{\infty}(\Omega_{\epsilon}\times(0,T))}^2
        +\norm{p_{\phi}}_{L^{\infty}(\Omega_{\epsilon}\times(0,T))}\right]\\
        &\hspace{10pt}\times\left[\delta\norm{\mathbf{u}_{\omega}}_{L^2((0,T),L^{\infty}(\Omega_{\epsilon}))} 
        + \norm{\mathbf{n}\bdot\mathbf{u}_{\omega}}_{L^2((0,T),L^{\infty}(\Omega_{3h}))}\right]
    \end{aligned}
\end{equation}
\begin{equation}
    \begin{aligned}
        &\Bigg|\int_0^T\int_{\Omega}\varphi\theta_{h,\ell}'\bn\bdot(\tau_{\ell}(\bu_{\omega},\bu_{\omega}) +\tau_{\ell}(\bu_{\phi},\bu_{\omega})+\tau_{\ell}(\bu_{\omega},\bu_{\phi}))\bdot\Bar{\mathbf{u}}_{\phi,\ell}\,dV\, dt\Bigg|\\
        &\lesssim\norm{\psi}_{L^{\infty}((\partial B)_T)}\times
        \bigg[
        \left(\delta\norm{\mathbf{u}_{\omega}}_{L^2((0,T),L^{\infty}(\Omega_{\epsilon}))} 
        + \norm{\mathbf{n}\bdot\mathbf{u}_{\omega}}_{L^2((0,T),L^{\infty}(\Omega_{3h}))}\right)\\
        &\hspace{20pt} \times\norm{\mathbf{u}_{\phi}}_{L^{\infty}(\Omega_{\epsilon}\times(0,T))}\times
        \left(\norm{\mathbf{u}_{\omega}}_{L^2((0,T),L^{\infty}(\Omega_{\epsilon}))} + \norm{\mathbf{u}_{\phi}}_{L^{\infty}(\Omega_{\epsilon}\times(0,T))}\right)\\
        &\hspace{80pt}+\left(\delta\norm{\mathbf{u}_{\phi}}_{L^{\infty}(\Omega_{\epsilon}\times(0,T))} + \norm{\mathbf{n}\bdot\mathbf{u}_{\phi}}_{L^{\infty}(\Omega_{3h}\times (0,T))}\right)\\
        &\hspace{120pt}\times
        \norm{\mathbf{u}_{\omega}}_{L^2((0,T),L^{\infty}(\Omega_{\epsilon}))} \norm{\mathbf{u}_{\phi}}_{L^{\infty}(\Omega_{\epsilon}\times(0,T))}
        \bigg]
    \end{aligned}
\end{equation}
\begin{equation}
    \begin{aligned}
        \Bigg|\int_0^T\int_{\Omega}\varphi\theta_{h,\ell}'\Bar{p}_{\omega}(\Bar{\mathbf{u}}_{\phi,\ell}\bdot\bn)\, dV\, dt\Bigg|
        &\lesssim\norm{\psi}_{L^{\infty}((\partial B)_T)}\times\norm{p_{\omega}}_{L^1((0,T), L^{\infty}(\Omega_{\epsilon}))}\\
        &\times\left(\delta\norm{\mathbf{u}_{\phi}}_{L^{\infty}(\Omega_{\epsilon}\times(0,T))} + \norm{\mathbf{n}\bdot\mathbf{u}_{\phi}}_{L^{\infty}(\Omega_{3h}\times (0,T)}\right)
    \end{aligned}
    \label{Jphi-bd4} 
\end{equation}
Thus, by the assumptions on the near-wall boundedness of $\mathbf{u}_{\omega}, p_{\omega}$ (\ref{wallboundedness}), the no-flow-through condition (\ref{wallnormal}) for $\bv_\omega$,
and corresponding properties of $\bv_\phi,$
\be \lim_{h,\ell\to 0}
\int_0^T \int_\Omega \varphi \grad\eta_{h,\ell}\bdot\Bar{\mathbf{J}}_{\phi,\ell}\,dV\, dt=0. \ee 
from which we conclude that 
\begin{align}
    \lim_{h,\ell\to0}\ext^*(\grad\eta_{h,\ell}\bdot\Bar{\mathbf{J}}_{\phi,\ell}) = 0 \quad \text{ in } D'((\partial B)_T)
\end{align}
and thus $\bu_{\phi}\bdot\btau_w = 0$ by Theorem \ref{theorem2}. More generally, all of 
the above bounds are valid for any $\varphi\in\Bar{D}(\Bar{\Omega}\times(0,T))$ with 
$\psi:=\left.\varphi\right|_{(\partial B)_T}$ and the local balance equation 
(\ref{weakInteractionLimit2}) then follows directly from \eqref{weakInteractionLimit1}
in Theorem \ref{theorem1}. Finally, with the further assumption \eqref{L2Conv-str} of 
global $L^2$-convergence, then \eqref{dJA-zero} of Theorem \ref{theorem1} holds and 
reduces to \eqref{dJA-zero2} since $\bu_{\phi}\bdot\btau_w = 0.$ 

\section{Proof of Theorems \ref{theorem4}-\ref{theorem6}}\label{thm46}

\subsection{Proof of Theorem \ref{theorem4}}
% TODO: Need L^{\infty} bound for p_{\omega}
% TODO: this should be discussed before theorem 4
From the uniform bound \eqref{pL32Bd}, we can use the same type of diagonal argument as 
in \cite{quan2022inertial}, Remark 2, to extract a subsequence $\nu_j\to 0$ so that 
$p^{\nu_j}_\omega\to p_\omega\in L^{\frac{3}{2}}((0,T),L^{\frac{3}{2}}(\Omega))$
distributionally, with $(\bv_\omega,p_\omega)$ a weak solution of \eqref{E-omega-mom2}. 
We then take an arbitrary $\varphi\in\Bar{D}(\Bar{\Omega}\times(0,T))$ with $\psi = \varphi|_{\partial B}$, and test 
the energy equation (\ref{Eom-loc}) of the rotational flow to obtain:  
\begin{equation}
    \begin{aligned}
        -\int_0^T\int_{\Omega}&\frac{1}{2}\partial_t\varphi|\mathbf{u}_{\omega}^{\nu}|^2
        \,dV\,dt
        -\int_0^T\int_{\Omega}\grad\varphi\bdot\left[\frac{1}{2}|\mathbf{u}_{\omega}^{\nu}|^2\mathbf{u}^{\nu}+p_{\omega}^{\nu}\mathbf{u}_{\omega}^{\nu}\right]dV\,dt\\
        &=\int_0^T\int_{\Omega}\varphi\grad\bdot(\nu\mathbf{u}_{\omega}^{\nu}\btimes\bomega)
        \,dV\,dt
        -\int_0^T\int_{\Omega}\nu\varphi|\bomega^{\nu}|^2\,dV\,dt\\
        &-\int_0^T\int_{\Omega}\varphi\grad\mathbf{u}_{\phi}\bdots\mathbf{u}_{\omega}^{\nu}\otimes\mathbf{u}_{\omega}^{\nu}\,dV\,dt
    \end{aligned}\label{rotationTest}
\end{equation}
\iffalse
Assumption (\ref{viscDissLimit2}) immediately gives
\begin{align}
    \int_0^T\int_{\Omega}\nu\varphi|\bomega^{\nu}|^2d\mathbf{x}dt 
    = \int_0^T\int_{\Omega}\varphi Q^{\nu}d\mathbf{x}dt
    \xrightarrow{\nu\to0}\int_0^T\int_{\Omega}\varphi Qd\mathbf{x}dt
\end{align}
\fi
Note that
\begin{equation}
    \begin{aligned}
        \int_0^T\int_{\Omega}&\varphi\grad\bdot(\nu\mathbf{u}_{\omega}^{\nu}\btimes\bomega^{\nu})
        \,dV\,dt\\ 
        &= -\int_0^T\int_{\Omega}\varphi\grad\bdot(\nu\mathbf{u}_{\phi}\btimes\bomega^{\nu})
        \,dV\,dt 
        + \int_0^T\int_{\Omega}\varphi\grad\bdot(\nu\mathbf{u}^{\nu}\btimes\bomega^{\nu})
        \,dV\,dt
    \end{aligned}
\end{equation}
Because of the identity $\grad\bdot(\bv^\nu\btimes\bomega^\nu)=\grad\otimes\grad\bdots(\bv^\nu\otimes\bv^\nu)
-\triangle\left(\frac{1}{2}|\bv^\nu|^2\right),$ we have by integration by parts and the no-slip condition on $\bv^\nu$ that 
\begin{align}
    &\int_0^T\int_{\Omega}\varphi\grad\bdot(\nu\mathbf{u}^{\nu}\btimes\bomega^{\nu})
    \,dV\,dt
    \\\nonumber 
    &= -\nu \int_0^T\int_{\Omega}
    \left[ (\grad\otimes\grad)\varphi\bdots\bv^{\nu}\otimes\bv^{\nu}
    -(\triangle\varphi)\,\frac{1}{2}|\bv^\nu|^2\right]dV\,dt
\end{align}
It follows that 
\begin{eqnarray}  
    \left|\int_0^T\int_{\Omega}\varphi\grad\bdot(\nu\mathbf{u}^{\nu}\btimes\bomega^{\nu})
    \,dV\,dt
    \right|
    &\lesssim&  \nu \norm{\grad\otimes\grad\varphi}_{L^\infty((0,T)\times\Omega)} \norm{\bv^\nu}^2_{L^2((0,T),L^2(K_\varphi))} \cr
    & \xrightarrow{\nu\to0}& 0 
\end{eqnarray}

Next, from Theorem \ref{theorem1}, we have
\be
    \int_0^T\int_{\Omega} \varphi\grad\bdot(\nu\mathbf{u}_{\phi}\btimes\bomega^{\nu})\,dV\,dt
    \xrightarrow{\nu\to0}-\langle\bu_{\phi}\bdot\btau_w,\psi\rangle
\ee 
The convergence of all of the rest of the terms in the equation (\ref{rotationTest}) can be easily deduced by Lemma 1 
in \cite{quan2022inertial} in conjunction with the assumptions (\ref{L3Conv}-\ref{pL32Boundary}), except for the 
term involving $Q^\nu.$ We thus deduce from (\ref{rotationTest}) that the limit 
\eqref{viscDissLimit2} of $\langle Q^\nu,\varphi\rangle$ must exist as well with 
\begin{eqnarray} 
        && %-\int_{\Omega}\frac{1}{2}\varphi|\mathbf{u}_{\omega}|^2(0,x)d\mathbf{x}
        \langle Q,\varphi\rangle :=
        \int_0^T\int_{\Omega}\partial_t\varphi\cdot \frac{1}{2}|\mathbf{u}_{\omega}|^2\,dV\, dt
        \int_0^T\int_{\Omega}\grad\varphi\bdot\left[\frac{1}{2}|\mathbf{u}_{\omega}|^2\mathbf{u}+p_{\omega}\mathbf{u}_{\omega}\right]
        dV\,dt\cr 
        && \hspace{50pt} +\langle\bv_\phi\bdot \btau_w,\psi\rangle
        -\int_0^T\int_{\Omega}\varphi\grad\mathbf{u}_{\phi}\bdots\mathbf{u}_{\omega}\otimes\mathbf{u}_{\omega}\,dV\,dt.
        \label{Qdef}
    \end{eqnarray} 
From this definition, the balance equation (\ref{fgRotationalLimit2}) trivially follows. Furthermore,
$Q$ defined by \eqref{Qdef} is clearly a linear functional on $\Bar{D}(\Bar{\Omega}\times (0,T)).$ 
Finally, since $\langle Q^\nu,\varphi\rangle\geq 0$ for all $\nu>0$ when $\varphi\geq 0,$ then 
the limit functional also satisfies the corresponding inequality $\langle Q,\varphi\rangle\geq 0.$ Thus, $Q$ is non-negative. 

\subsection{Proof of Theorem \ref{theorem5}}
We take an arbitrary $\varphi\in\Bar{D}(\Bar{\Omega}\times(0,T))$ with $\psi = \varphi|_{\partial B}$.\\
Testing the inviscid coarse-grained balance equation \eqref{cgwEulerRotationalEq2} for energy 
in rotational flow against $\varphi$ and rearranging yields
\begin{equation}
    \begin{aligned}
        &\int_0^T\int_{\Omega}\varphi\grad\eta_{h,\ell}\bdot\Bar{\bJ}_{\omega, \ell}\,dV\,dt+\int_0^T\int_{\Omega}\varphi\eta_{h,\ell}\grad\Bar{\mathbf{u}}_{\omega,\ell}
        \bdots \tau_{\ell}(\bu_{\omega},\bu_{\omega})\,dV\,dt\\
        &=- \int_0^T\int_{\Omega}\partial_t\varphi\cdot\frac{1}{2}\eta_{h,\ell}|\Bar{\mathbf{u}}_{\omega,\ell}|^2 \,dV\,dt
        -\int_0^T\int_{\Omega}\grad\varphi\bdot\eta_{h,\ell}\Bar{\bJ}_{\omega, \ell}\, dV\,dt\\
        &\hspace{10pt} +\int_0^T\int_{\Omega}\varphi\eta_{h,\ell}\grad\Bar{\mathbf{u}}_{\phi,\ell}
        \bdots \Bar{\mathbf{u}}_{\omega,\ell}\otimes\Bar{\mathbf{u}}_{\omega,\ell}\,dV\,dt\\
        &\hspace{10pt} -\int_0^T\int_{\Omega}\varphi\eta_{h,\ell}\grad\Bar{\mathbf{u}}_{\omega,\ell}
        \bdots(\tau_{\ell}(\bu_{\phi},\bu_{\omega})+\tau_{\ell}(\bu_{\omega},\bu_{\phi}))\,dV\,dt
    \end{aligned}\label{cgEulerRotationalEq}
\end{equation}
Let $K_{\varphi}\subset\mathbb{R}^3$ be a compact set such that
$
    \text{supp}(\varphi)\subset K_{\varphi}\times(0,T)\subset\Bar{\Omega}\times(0,T)
$
Similar to the proof of Theorem \ref{theorem2}, one can use Lemma \ref{lemma1} to obtain 
\begin{align}
    \frac{1}{2}\eta_{h,\ell}|\Bar{\mathbf{u}}_{\omega,\ell}|^2\xrightarrow[L^{\frac{3}{2}}((0,T),L^{\frac{3}{2}}(K_{\varphi}))]{h,\ell\to0}\frac{1}{2}|\mathbf{u}_{\omega}|^2, \;\;\;\; &\eta_{h,\ell}\Bar{p}_{\omega, \ell}\Bar{\mathbf{u}}_{\omega,\ell}\xrightarrow[L^1((0,T),L^1(K_{\varphi}))]{h,\ell\to0}p_{\omega}\mathbf{u}_{\omega}\label{thm5-conv1}
\end{align}
and
\begin{align}
    \frac{1}{2}\eta_{h,\ell}|\Bar{\mathbf{u}}_{\omega,\ell}|^2\Bar{\mathbf{u}}
    &\xrightarrow[L^1((0,T),L^1(K_{\varphi}))]{h,\ell\to0}\frac{1}{2}|\mathbf{u}_{\omega}|^2\mathbf{u}\label{thm5-conv2}\\
    \eta_{h,\ell}\grad\Bar{\mathbf{u}}_{\phi,\ell}\bdots\Bar{\mathbf{u}}_{\omega,\ell}\otimes\Bar{\mathbf{u}}_{\omega,\ell}
    &\xrightarrow[L^1((0,T),L^1(K_{\varphi}))]{h,\ell\to0}\grad\mathbf{u}_{\phi}\bdots\mathbf{u}_{\omega}
    \otimes\mathbf{u}_{\omega}\label{thm5-conv3}\\
    \Bar{\bJ}_{\omega,l}
    &\xrightarrow[L^1((0,T),L^1(K_{\varphi}))]{h,\ell\to0}\frac{1}{2}|\bu_{\omega}|^2\bu + p_{\omega}\bu_{\omega}\label{thm5-conv4}
\end{align}

Next we note that 
\begin{align}
    \lim_{h,\ell\to0}\int_0^T\int_{\Omega}\varphi\eta_{h,\ell}\grad\Bar{\mathbf{u}}_{\omega,\ell}
    \bdots\tau_{\ell}(\bu_{\omega},\bu_{\phi})\, dV\,dt = 0\label{thm5-conv5}
\end{align}
\begin{align}
    \lim_{h,\ell\to0}\int_0^T\int_{\Omega}\varphi\eta_{h,\ell}\grad\Bar{\mathbf{u}}_{\omega,\ell}
    \bdots\tau_{\ell}(\bu_{\phi},\bu_{\omega})\,dV\,dt = 0\label{thm5-conv6}
\end{align}
by a standard density argument, based on the Constantin-E-Titi identities \cite{constantin1994onsager}: 
\begin{align}
    \grad\Bar{\bu}_{\omega,\ell}(\bx) = -\frac{1}{\ell}\langle (\grad G)_{\ell}\delta \bu_{\omega}(\cdot;\bx)\rangle 
\end{align}
\begin{align}
    \tau_{\ell}(\bu_{\omega},\bu_{\phi})(\bx) = \langle G_{\ell}\delta\bu_{\omega}(\cdot;\bx)\otimes\delta\bu_{\phi}(\cdot;\bx)\rangle
    - \langle G_{\ell}\delta\bu_{\omega}(\cdot;\bx)\rangle \otimes \langle G_{\ell}\delta\bu_{\phi}(\cdot;\bx)\rangle\label{cet-decomposition}
\end{align}
where we use $\langle\cdot\rangle$ as an abbreviated notation for spatial integral over $\mathbb{R}^3$.
We give here a few details, discussing only the term 
\begin{align}
    \frac{1}{\ell}\int_0^T\int_{\Omega}\eta_{h,\ell}\varphi
    \langle (\grad G)_{\ell}\delta\bu_{\omega}(\cdot;\bx)\rangle \bdots 
    \langle G_{\ell}\delta\bu_{\omega}(\cdot;\bx)\otimes\delta\bu_{\phi}(\cdot;\bx)\rangle
    dV\,dt,\label{thm5-integral}
\end{align}
as all others can be treated in the same way. Since $\mathbf{u}_{\omega}\in L^3((0,T),L^3(K_{\varphi}))$
from \eqref{L3Conv},\eqref{uL3Boundary}, then for all $\epsilon>0$, there exists 
$\mathbf{v}^{\epsilon}\in C^{\infty}((0,T)\times \text{Int}(K_{\varphi}))$ such that $\norm{\mathbf{u}_{\omega}-\mathbf{v}^{\epsilon}}_{L^3((0,T),L^3(K_{\varphi}))}<\epsilon$. 
Defining $\mathbf{\Delta}^{\epsilon} := \mathbf{u}_{\omega}-\mathbf{v}^{\epsilon}$,
we can then substitute $\mathbf{u}_{\omega}=\mathbf{v}^\epsilon+\mathbf{\Delta}^\epsilon$ into 
\eqref{thm5-integral} to obtain two integrals. The first satisfies the bound 
\begin{equation}
    \begin{aligned}
        &\left|\frac{1}{\ell}\int_0^T\int_{\Omega}\eta_{h,\ell}\varphi
        \langle (\grad G)_{\ell}\delta\bu_{\omega}(\cdot;\bx)\rangle)\bdots 
        \langle G_{\ell}\delta\mathbf{v}^{\epsilon}(\cdot;\bx)\otimes\delta\bu_{\phi}(\cdot;\bx)\rangle
        dV\,dt
        \right|\\
        &\lesssim \ell\norm{\varphi}_{L^{\infty}(\mathbb{R}^3\times(0,T))}
        \norm{\mathbf{u}_{\omega}}_{L^{3}((0,T),L^{3}(K_{\varphi}))}
        \norm{\grad\mathbf{v}^{\epsilon}}_{L^{\infty}(K_{\varphi}\times(0,T))}\norm{\grad\mathbf{u}_{\phi}}_{L^{\infty}(K_{\varphi}\times(0,T))}
    \end{aligned}\label{boundA}
\end{equation}
where the factor $\ell$ is obtained from the smoothness of $\mathbf{v}^{\epsilon}$ and $\bu_{\phi}$. 
Here, $\lesssim$ is inequality with prefactor depending on $K_{\varphi}\times(0,T)$ and the mollifier $G$. 
Thus, the first integral vanishes as $h,\ell\to0$, with $\epsilon$ fixed. The second satisfies 
\begin{equation}
    \begin{aligned}
        &\left|\frac{1}{\ell}\int_0^T\int_{\Omega}\eta_{h,\ell}\varphi
        \langle (\grad G)_{\ell}\delta\bu_{\omega}(\cdot;\bx)\rangle\bdots
        \langle G_{\ell}\delta\mathbf{\Delta}^{\epsilon}(\cdot;\bx)\otimes\delta\bu_{\phi}(\cdot;\bx)\rangle
        dV\,dt\right|\\
        &\lesssim\epsilon\norm{\varphi}_{L^{\infty}(\mathbb{R}^3\times(0,T))}
        \norm{\mathbf{u}_{\omega}}_{L^{3}((0,T),L^{3}(K_{\varphi}))}
        \norm{\grad\mathbf{u}_{\phi}}_{L^{\infty}(K_{\varphi}\times(0,T))}
    \end{aligned}\label{boundB}
\end{equation}
for any $\epsilon>0.$  Thus, we obtain that (\ref{thm5-integral}) vanishes in the limit of $h,\ell\to0$. 

Since the terms considered in (\ref{thm5-conv1})-(\ref{thm5-conv4}) all converge in $L^1((0,T),L^1(K_{\varphi}))$, 
we conclude from \eqref{cgEulerRotationalEq} that 
\begin{equation}
    \begin{aligned}
        &\lim_{h,\ell\to0}\int_0^T\int_{\Omega}\varphi\left(\grad\eta_{h,\ell}\bdot\Bar{\bJ}_{\omega, \ell} + \eta_{h,\ell}\grad\Bar{\mathbf{u}}_{\omega,\ell}\bdots\tau_{\ell}(\bu_{\omega},\bu_{\omega})\right)
        dV\,dt\\
        =&- \int_0^T\int_{\Omega}\partial_t\varphi\cdot \frac{1}{2}|\mathbf{u}_{\omega}|^2 \,dV\,dt
        +\int_0^T\int_{\Omega}\varphi\grad\mathbf{u}_{\phi}\bdots\mathbf{u}_{\omega}\otimes\mathbf{u}_{\omega}\,dV\,dt\\
        &\hspace{30pt} -\int_0^T\int_{\Omega}\grad\varphi\bdot\left[\frac{1}{2}|\mathbf{u}_{\omega}|^2\mathbf{u}
        +p_{\omega}\mathbf{u}_{\omega}\right]dV\,dt
    \end{aligned}
\end{equation}
for all $\varphi\in\Bar{D}(\Bar{\Omega}\times(0,T))$.
By comparison with the inviscid energy balance equation (\ref{fgRotationalLimit2}) for 
rotational flow  in Theorem \ref{theorem4}, we obtain (\ref{Jom-bal}).

\subsection{Proof of Theorem \ref{theorem6}} The near-wall boundedness property (\ref{wallboundednessL3}) and no-flow-through condition (\ref{wallnormalL3}) imply those in Theorem \ref{theorem3}, which immediately gives 
\begin{align}
    \bu_{\phi}\bdot\btau_w = 0 \quad \mbox{ in $D'((\partial B)_T).$} 
\end{align}
Thus, the local balance equation for rotational energy (\ref{fgRotationalLimit2}) reduces to
\begin{eqnarray} 
    && %-\int_{\Omega}\frac{1}{2}\varphi|\mathbf{u}_{\omega}|^2(0,x)d\mathbf{x}
    -\int_0^T\int_{\Omega}\frac{1}{2}\partial_t\varphi|\mathbf{u}_{\omega}|^2\,dV\, dt
    -\int_0^T\int_{\Omega}\grad\varphi\bdot\left[\frac{1}{2}|\mathbf{u}_{\omega}|^2\mathbf{u}+p_{\omega}\mathbf{u}_{\omega}\right]
    dV\,dt\cr 
    && \hspace{10pt} =
    -\int_0^T\int_{\Omega}\varphi Q \,dV\,dt -\int_0^T\int_{\Omega}\varphi\grad\mathbf{u}_{\phi}\bdots\mathbf{u}_{\omega}\otimes\mathbf{u}_{\omega}\,dV\,dt
\end{eqnarray}

Using the same arguments as in the proof of Theorem \ref{theorem3}, one can obtain upper bounds on the spatial flux of kinetic energy of the rotational flow, defined as
\begin{equation}
    \begin{aligned}
        \grad\eta_{h,\ell}\bdot\Bar{\mathbf{J}}_{\omega, \ell}(\bx) &= \theta'_{h,\ell}(d(\mathbf{x}))\bigg(\frac{1}{2}|\Bar{\mathbf{u}}_{\omega,\ell}|^2
        \Bar{\mathbf{u}}_\ell\bdot\mathbf{n}(\pi(\mathbf{x}))+\Bar{p}_{\omega, \ell}\Bar{\mathbf{u}}_{\omega,\ell}\bdot\mathbf{n}(\pi(\mathbf{x}))\\
        &+\mathbf{n}(\pi(\mathbf{x}))\bdot(\tau_{\ell}(\bu_{\omega},\bu_{\omega})+\tau_{\ell}(\bu_{\phi},\bu_{\omega})+\tau_{\ell}(\bu_{\omega},\bu_{\phi}))\bdot\Bar{\mathbf{u}}_{\omega,\ell}\bigg), 
    \end{aligned}
    \label{Jom-flux}
\end{equation}
which are analogous to the bounds \eqref{Jphi-bd1}-\eqref{Jphi-bd4} on the spatial flux of 
the interaction energy. We need estimates $L^3$ in time here since \eqref{Jom-flux} is cubic in $\bv_\omega.$
It then follows from these upper bounds and the assumptions (\ref{wallboundednessL3}), 
\eqref{wallnormalL3} that 
\begin{align}
    \lim_{h,\ell\to0}\int_0^T\int_{\Omega}\varphi\grad\eta_{h,\ell}\bdot\Bar{\bJ}_{\omega, \ell} 
    \, dV\,dt = 0
\end{align}
for all $\varphi\in\Bar{D}(\Bar{\Omega}\times(0,T))$ and this implies that for all $\ext\in\cE$,
\begin{align}
    \lim_{h,\ell\to0}\ext^*[\grad\eta_{h,\ell}\bdot\Bar{\bJ}_{\omega,\ell}] = 0.
\end{align}
By comparison with (\ref{Jom-bal}), we obtain the equality (\ref{roughnessDissipative}) between inertial energy dissipation and zero-viscosity limit of viscous energy dissipation.

\begin{acknowledgements} 
We are grateful to Samvit Kumar, Charles Meneveau, and Tamer Taki for discussions of this problem. 
This work was funded by the Simons Foundation, via Targeted Grant in MPS-663054.
\end{acknowledgements}

\bibliographystyle{plain}
\bibliography{bibliography}

\end{document}